\documentclass[runningheads]{llncs}
\usepackage[T1]{fontenc}
\usepackage{graphicx}
\usepackage{comment}
\usepackage{color}
\usepackage{algorithm}

\usepackage{algorithmicx}
\usepackage[noend]{algpseudocode}

\usepackage{pgf,tikz} 

\usepackage{amsmath} 
\usepackage{amssymb}
\usepackage{bm}
\usepackage{enumitem}

\usepackage[dvipsnames]{xcolor}
\usepackage[suppress]{color-edits}
\addauthor[suppress]{pa}{orange}
\addauthor[suppress]{ar}{blue}
\addauthor[journal]{pj}{red}

\newcommand{\adapt}{{\sc Adaptive}}
\newcommand{\adaptRec}{{\sc AdaptiveWithRecourse}}
\newcommand{\nonadapt}{{\sc NonAdaptive}}

\makeatletter
\algnewcommand{\LineComment}[1]{\Statex \hskip\ALG@thistlm \(\triangleright\) #1}
\makeatother

\def\cra{{\sc cr}}
\newtheorem{observation}{Observation}
\newtheorem{ourClaim}{Claim}
\begin{document}

\newcommand{\algorithmicbreak}{\textbf{break}}
\newcommand{\BREAK}{\State \algorithmicbreak}

\title{\pjreplace{Removable Online Knapsack: Exploiting Recourse and Bounded Item Sizes}{Leveraging Recourse and Size Bounds for Online Knapsack with Removable Items}\thanks{\pjcomment{They wanted this reference somewhere.}\pjedit{A preliminary version of this work appeared in IWOCA 2026~\cite{FGPP26}.}}}
%
\titlerunning{Leveraging Recourse and Size Bounds for Online Knapsack with Removable Items}
%


\author{Dimitris Fotakis\inst{1,2}\orcidID{0000-0001-6864-8960} \and
Laurent Gourvès \inst{3}\orcidID{0000-0002-5076-1583} \and
Aris Pagourtzis\inst{1,2}\orcidID{0000-0002-6220-3722} \and
Panagiotis Patsilinakos \inst{3}\orcidID{0009-0007-9540-2334}}


%
\authorrunning{D. Fotakis et al.}
%

\institute{National Technical University of Athens, Greece \and
Archimedes, Athena Research Center, Greece  \and
Université Paris-Dauphine, Université PSL, CNRS, LAMSADE,
Paris, France  
\email{fotakis@cs.ntua.gr}, 
\email{laurent.gourves@dauphine.fr},\\
\email{pagour@cs.ntua.gr},  
 \email{panagiotis.patsilinakos@dauphine.psl.eu}}


%
\maketitle              
\begin{abstract}
In the online unweighted knapsack problem, items arrive sequentially and must be accepted or rejected immediately for a bounded-capacity knapsack, aiming to maximize the total size of accepted items. When decisions are irrevocable, no algorithm achieves a 
non-trivial  
competitive ratio. An essential relaxation that circumvents this impossibility is to allow previously accepted items to be removed from the knapsack (a.k.a. removal). In this work, we investigate possible further improvements on the competitive ratio with removal by bringing together two other previously proposed model extensions: knowledge of an upper bound on the item sizes (a.k.a. bound awareness) and ability to reinsert previously rejected items (a.k.a. recourse). Our contributions are twofold. First, we establish close upper and lower bounds on the competitive ratio under bound awareness, removal and recourse. We show that significant improvements on the best known competitive ratio are possible when the upper bound on the item sizes is not too restrictive and only few uses of recourse are available. Second, we introduce the notion of adaptive online algorithms 
which do not require to know any upper bound on the item sizes but still achieve performance guarantees comparable to the ones obtained under the bound awareness model.


\end{abstract}
\section{Introduction}


We study the online unweighted knapsack problem 
in which some 
items, to be placed in a knapsack of capacity $1$, are disclosed over time. 
Each item $o_i$ has a \emph{size} $|o_i| \in (0,1]$. Starting from an initial empty set $S$,  item $o_i$ is revealed at round $i$ during which it must be decided to pack $o_i$ in $S$ or not. 
The value of $S$, denoted by $v(S)$ and defined as $\sum_{o_i \in S} |o_i|$, should never exceed the capacity. The goal is to  maximize $v(S)$ when the instance stops, given that the number of items is not known in advance. 

The competitive ratio ({\cra} for short) $\rho$ of a deterministic online algorithm is the worst case ratio between the value of its solution $\hat S$ and the value of an optimal solution $S^*$: 
$\rho := {v(\hat S)}/{v(S^*)}.$ Every {\cra} of this article is in  $[0,1]$. 
Unfortunately, no deterministic algorithm can exhibit a non-trivial    
{\cra} (i.e., bounded away from zero by any constant)\footnote{Suppose a first item $o_1$ 
is released. If $o_1$ is not packed, then the instance stops. Otherwise, a second item $o_2$ of size $1$ arrives and the instance stops. Clearly,  $o_2$ cannot be packed because of $o_1$. The {\cra} of any deterministic algorithm for this instance is at most $|o_1|=\max(0/|o_1|,|o_1|/1)$ but $|o_1|$ can be arbitrarily close to $0$.} for the online unweighted knapsack problem \cite{MSV95}. This is mostly due to the fact that decisions are irrevocable. An essential relaxation that circumvents this impossibility is to allow items to be removed from $S$ after their insertion. Iwama and Taketomi showed that the online unweighted knapsack problem with removable items admits a $\phi^{-1}$-competitive algorithm with $\phi^{-1}:=\frac{\sqrt{5}-1}{2} \approx 0.618$,     
and that no deterministic online algorithm can achieve a better {\cra} \cite{IT02}.  

Refined models have been proposed for improving upon the $\phi^{-1}$-competitive algorithm.  
One such model is \emph{recourse}, in which a recourse operation allows one item that was previously rejected to be eventually packed. Concretely, an item may be discarded from the solution upon its arrival and temporarily stored in a buffer. Subsequently, this item may eventually be inserted into the solution, which requires a so-called recourse operation. \pjmargincomment{For completeness}\pjedit{Items removed after being packed can also be stored in the buffer and reinserted by a recourse operation.}

In contrast to the unlimited number of removals, the number of recourse operations is typically small, as such operations can be costly in practice. B{\"{o}}ckenhauer \emph{et al.} have demonstrated (among other results) that the online unweighted knapsack problem  with at most $r \in \{1,2,3\}$ recourse operation(s) 
and removable\footnote{Recourse operations are useless without the removability property \cite{BKMRSW26}.} items admits a deterministic $\frac{r+1}{r+2}$-competitive   algorithm,  and that no algorithm can achieve better {\cra} with $r$ recourse operation(s) \cite{BKMRS23,BKMRSW26}.  Consequently, a single recourse operation suffices to improve upon the ratio of $\phi^{-1}$.

The second model assumes that item sizes are upper-bounded by a value strictly smaller than the knapsack capacity. Moreover, when dealing repeatedly with such bounded instances, this bound can be known in advance and with precision. 

Accordingly, \emph{bound awareness} refers to the ability to know from the outset an upper bound $u \in (0,1]$ such that $|o_i| \le u$ holds for every item $o_i$. 
As $u$ approaches $0$, the instance becomes increasingly constrained; however, the potential errors of an algorithm have a diminishing impact on the {\cra}. 
When items are removable and their sizes are upper bounded by $u$, Gourv{\`e}s and Pagourtzis \cite{GourvesP24} proposed an algorithm parameterized by $u$ whose {\cra} $\rho(u)$  improves upon $\phi^{-1}$ whenever $u \le \phi^{-1}$, given that $\phi^{-1}$ is known to be the optimal {\cra} (without recourse) when $u \in (\phi^{-1},1]$ \cite{IT02}. Figure \ref{fig4} comprises an illustration of $\rho(u)$  while its formal definition is given in Theorem \ref{thm11}.

In comparison with the $\frac{r+1}{r+2}$-competitive algorithm of  B{\"{o}}ckenhauer \emph{et al.}, $\rho(u)$ exceeds $2/3$ for $u< 0.6$, exceeds $3/4$ for $u< 0.4$, and exceeds $4/5$ for $u< 5/17$. 

The approaches of \cite{BKMRSW26} and \cite{GourvesP24} are complementary  yet incomparable: the former considers recourse without bound awareness, while the latter assumes bound awareness but allows no recourse.

\subsection{Contribution}

Our objective in this article is to extend the results of \cite{BKMRSW26} and \cite{GourvesP24}.  

To this end, we explore in Section \ref{sec:recourse} the possibility of combining recourse with bound awareness for the online  unweighted knapsack problem with removable items. We establish both upper and lower bounds on the optimal {\cra}, given an upper bound $u$ on item sizes and a fixed limit on the number of recourse operations. Our  bounds, established for many discrete values of $u$ (typically $u=1/k$ for integers $k$), increase as 
$u=1/k$ decreases, 
as does the number of possible recourse operations (on the order of $k/2$ for the lower bound and of $k$ for the upper bound). Our most significant results concern the most relevant cases in which $u$ is large (so the instance is only weakly constrained) and the number of possible recourse operations is small.

Figure \ref{fig4} provides an illustration where bullets (resp.,  triangles) correspond to lower (resp., upper) bounds on the optimal {\cra} for a small number of recourse operations (this number is indicated in parentheses, to the right of the bullet/triangle). 
Concretely, having  $\blacktriangle (r)$ at coordinates $(x,y)$ means that for instances whose items have size at most $x$, every algorithm using at most $r$ recourse operations has {\cra} at most $y$. Having $\bullet (r)$ at coordinates $(x,y)$ means that there exists a $y$-competitive algorithm using at most $r$ recourse operations
for instances whose items have size at most $x$.

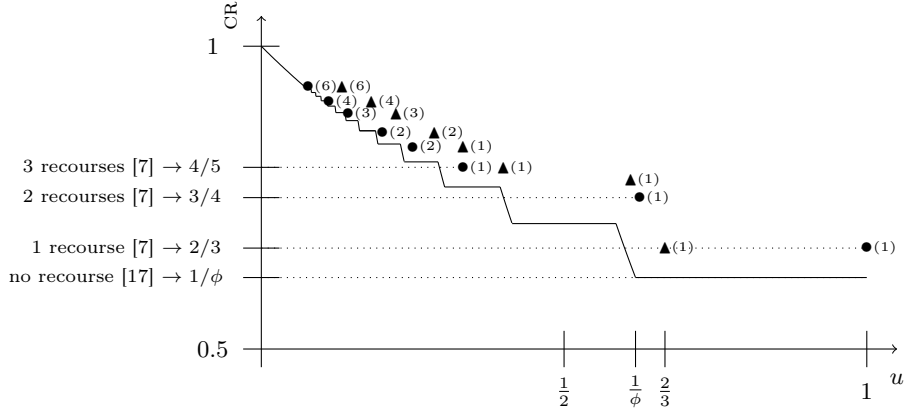
\begin{figure}
\center
\begin{tikzpicture}[scale=8]

\draw[->] (-0.03,0.5)--(1.05,0.5);
\draw[->] (0,0.47)--(0,1.05);

\draw  (1,0.47)--(1,0.53);
\draw  (2/3,0.47)--(2/3,0.53);
\draw  (0.03,1)--(-0.03,1);
\draw  (0.03,2/3)--(-0.03,2/3);
\draw  (0.03,3/4)--(-0.03,3/4);
\draw  (0.03,4/5)--(-0.03,4/5);

\draw (-0.03,0.618033989)--(0.03,0.618033989);
\draw (0.618033989,0.53)--(0.618033989,0.47);
\draw (0.5,0.53)--(0.5,0.47);


\draw[dotted]  (0,0.618033989)--(1,0.618033989);
\draw[dotted]  (0,2/3)--(1,2/3);
\draw[dotted]  (0,3/4)--(0.618033989,3/4);
\draw[dotted]  (0,4/5)--(1/3,4/5);

\node () at (1.05,0.45) {$u$};
\node [rotate=90] () at (-0.05,1.05) {{\cra}};
\node () at (-0.08,1) {$1$};
\node () at (0.618033989,0.43) { $\frac{1}{\phi}$};
\node () at (1,0.43) {$1$};
\node () at (2/3,0.43) {$\frac{2}{3}$};
\node () at (0.5,0.43) {$\frac{1}{2}$};
\node () at (-0.08,0.5) {$0.5$};

\draw plot [domain=0:0.076118446,variable=\t] (\t,{1/(1+\t)});

\draw plot [domain=0.585786438:0.618033989,variable=\t] (\t,{(1-\t)/\t});
\draw plot [domain=0.394448725:0.414213562,variable=\t] (\t,{(1-\t)/(2*\t)});
\draw plot [domain=0.291796068:0.302775638,variable=\t] (\t,{(1-\t)/(3*\t)});
\draw plot [domain=0.229670386:0.236067977,variable=\t] (\t,{(1-\t)/(4*\t)});
\draw plot [domain=0.192582404:0.188611699,variable=\t] (\t,{(1-\t)/(5*\t)});
\draw plot [domain=0.16227766:0.159670332,variable=\t] (\t,{(1-\t)/(6*\t)});
\draw plot [domain=0.140054945:0.138260621,variable=\t] (\t,{(1-\t)/(7*\t)});
\draw plot [domain=0.123105626:0.121822171,variable=\t] (\t,{(1-\t)/(8*\t)});
\draw plot [domain=0.108824378:0.109772229,variable=\t] (\t,{(1-\t)/(9*\t)});
\draw plot [domain=0.099019514:0.098300563,variable=\t] (\t,{(1-\t)/(10*\t)});
\draw plot [domain=0.090169944:0.089612167,variable=\t] (\t,{(1-\t)/(11*\t)});
\draw plot [domain=0.08276253:0.082321372,variable=\t] (\t,{(1-\t)/(12*\t)});
\draw plot [domain=0.076473219:0.076118446,variable=\t] (\t,{(1-\t)/(13*\t)});

\draw (0.618033989,0.618033989)--(1,0.618033989);
\draw (0.414213562,0.707106781)--(0.585786438,0.707106781);
\draw (0.394448725,0.767591879)--(0.302775638,0.767591879);
\draw (0.291796068,0.809016994)--(0.236067977,0.809016994);
\draw (0.229670386,0.838516481)--(0.192582404,0.838516481);
\draw (0.188611699,0.86037961)--(0.16227766,0.86037961);
\draw (0.188611699,0.86037961)--(0.16227766,0.86037961);

\draw (0.159670332,0.877150706)--(0.140054945,0.877150706);
\draw (0.138260621,0.890388203)--(0.123105626,0.890388203);
\draw (0.121822171,0.901085803)--(0.109772229,0.901085803);
\draw (0.108824378,0.909901951)--(0.099019514,0.909901951);
\draw (0.098300563,0.917288177)--(0.090169944,0.917288177);
\draw (0.089612167,0.923563544)--(0.08276253,0.923563544);
\draw (0.082321372,0.928959478)--(0.076473219,0.928959478);
\draw (0.076118446,0.933647701)--(0.071067812,0.933647701);

\node () at (-0.24,0.618033989) {\scriptsize no recourse \cite{IT02} $\rightarrow 1/\phi$};
\node () at (-0.223,2/3) {\scriptsize 1 recourse \cite{BKMRSW26} $\rightarrow 2/3$};
\node () at (-0.23,3/4) {\scriptsize $2$ recourses \cite{BKMRSW26} $\rightarrow 3/4$};
\node () at (-0.23,4/5) {\scriptsize 3 recourses \cite{BKMRSW26} $\rightarrow 4/5$};



\node () at (1/3,5/6) {{\scriptsize $\blacktriangle$}};
\node () at (1/3+0.03,5/6) {{\tiny $(1)$}};


\node () at (2/5,4/5) {{\scriptsize $\blacktriangle$}};
\node () at (2/5+0.03,4/5) {{\tiny $(1)$}};

\node () at (2/7,6/7) {{\scriptsize $\blacktriangle$}};
\node () at (2/7+0.03,6/7) {{\tiny $(2)$}};

\node () at (2/9,8/9) {{\scriptsize $\blacktriangle$}};
\node () at (2/9+0.03,8/9) {{\tiny $(3)$}};

\node () at (2/11,10/11) {{\scriptsize $\blacktriangle$}};
\node () at (2/11+0.03,10/11) {{\tiny $(4)$}};

\node () at (2/15,14/15) {{\scriptsize $\blacktriangle$}};
\node () at (2/15+0.03,14/15) {{\tiny $(6)$}};

\node () at (1/3,4/5) {$\bullet$};
\node () at (1/3+0.03,4/5) {{\tiny $(1)$}};

\node () at (1/4,5/6) {$\bullet$};
\node () at (1/4+0.03,5/6) {{\tiny $(2)$}};

\node () at (1/5,6/7) {$\bullet$};
\node () at (1/5+0.03,6/7) {{\tiny $(2)$}};

\node () at (1/7,8/9) {$\bullet$};
\node () at (1/7+0.03,8/9) {{\tiny $(3)$}};

\node () at (1/9,10/11) {$\bullet$};
\node () at (1/9+0.03,10/11) {{\tiny $(4)$}};

\node () at (1/13,14/15) {$\bullet$};
\node () at (1/13+0.03,14/15) {{\tiny $(6)$}};

\node () at (5/8,3/4) {$\bullet$};
\node () at (5/8+0.03,3/4) {{\tiny $(1)$}};

\node () at (1,2/3) {$\bullet$};
\node () at (1+0.03,2/3) {{\tiny $(1)$}};

\node () at (2/3,2/3) {{\scriptsize $\blacktriangle$}};
\node () at (2/3+0.03,2/3) {{\tiny $(1)$}};
\node () at (0.61,0.78) {{\scriptsize $\blacktriangle$}};
\node () at (0.61+0.03,0.78) {{\tiny $(1)$}};

\end{tikzpicture}
\caption{Bounds on the {\cra} as a function of $u$. The solid curve depicts $\rho(u)$, the {\cra} of the bound aware algorithm of \cite{GourvesP24}. 
Bullets and triangles correspond to lower and upper bounds on the optimal {\cra} when recourse is allowed (only drawn for a small number of recourse operations). The maximum number of recourse operations is indicated in parentheses. On the left are the {\cra}s from \cite{IT02,BKMRSW26}, which are independent of $u$. 
\label{fig4}}
\end{figure}

Our second approach, developed in Section \ref{sec:adaptive},  overcomes a drawback of bounded awareness. We continue to assume that item sizes are upper-bounded by a value $u$ 
but $u$ is \emph{not} known in advance. Therefore, we propose an \emph{adaptive} version of the parametric algorithm of 
\cite{GourvesP24}: its {\cra} is approximately the same, but the parameter $u$ is updated each time a (larger) new item is disclosed. Finally, an adaptive version of the algorithm that allows for one recourse operation is proposed, thereby improving the previously established competitive ratio.

Importantly, our adaptive algorithm yields a learning-augmented variant that can use an uncertain upper-bound prediction to match non-adaptive performance when correct, degrading by only a small factor in the worst case when incorrect.

All our bounds on {\cra}s are for deterministic algorithms for the online unweighted knapsack problem with removable items, so we often omit this information.  \pjdelete{Due to space constraints, some technical parts are omitted and deferred to the full version.}


\pjmargincomment{Had comparison with conference version here.}
\pjdelete{A preliminary version of this work appeared in IWOCA 2026~\cite{FGPP26}. This full version includes the proofs omitted from the conference version: the correctness of Algorithm~\ref{algorec4}, the algorithm and analysis for Theorem~\ref{thmk34}, the complementary upper bounds, and the analyses of the adaptive and learning-augmented algorithms. It also includes corrections to the adaptive algorithms and their analyses.}

 \subsection{Additional Related Work}

We have already mentioned the most related to our work earlier results. Some additional related work is described below.  

Regarding the possibility of item removal, an alternative model was proposed by Han \emph{et al.}  
\cite{HKM14}
where removal induces a cost either proportional to the item's size (by a value $f$ known as the {\em buyback factor}) or equal to~$1$.
The objective is to maximize the worth of packed items minus the cost of removed items, counting only those that were removed after been packed. Under this model, Han \emph{et al.} 
studied the unweighted online knapsack problem \cite{HKM14}, obtaining ratios that tend to $1/2$ when the buyback cost tends to $0$. Babaioff \emph{et al.}~\cite{BHK09} considered, among others, the online weighted knapsack problem with a known upper bound $u \le 1$ on the size of any item.
They gave a deterministic algorithm for the case $u<1/2$ whose competitive ratio is $(1-2u)/(1+2f+2\sqrt{f(1+f)})$, which becomes $1-2u$ when $f=0$. This is below the ratio  $\phi^{-1}$  of~\cite{IT02} except for very small values of $u$, and is also inferior to the ratios obtained in~\cite{GourvesP24} for every value of the item size bound $u$. However, the results of~\cite{BHK09} hold for the weighted case as well; note that no deterministic algorithm can achieve any non-trivial  
\cra{} for the online {\em weighted} knapsack, by removal only \cite{IwamaZ10}. 
Another result of~\cite{BHK09} is a randomized algorithm whose \cra{} tends to~$1/3$ as $f$ goes to $0$. 
Randomized algorithms for the online knapsack problem with (costless) removal can be found in \cite{HKM15,CJS16}.

A variant to the idea of recourse was proposed by B{\"{o}}ckenhauer \emph{et al.} \cite{BBHLR21}, where the possibility to reserve an item (i.e., postpone the decision about it) is studied. In that model items are non-removable and the decision on which to restore is taken at the end, assuming that the algorithm is informed once the instance ends and can take action afterwards. Assuming that  reserving  an item costs $\alpha$ times its value, 
the authors  characterize the best competitive ratio for all possible values of $\alpha\in (0,1)$. The best ratio  in~\cite{BBHLR21} is $1/2$, obtained for small values of $\alpha$ ($\le 1/6$), and the performance deteriorates as $\alpha$ increases. Note that in~\cite{BKMRSW26} they achieve a $2/3$ \cra{} with just one use of recourse and no information about the end of the instance; however, the two models are incomparable as in~\cite{BBHLR21} there is no limit on the number of items that can be restored from the reserve, while in~\cite{BKMRSW26} there is no cost for reserving items in the recourse and bringing them into the solution.
Another model supposes that the items can be placed in a buffer of size $K>1$ before a selection of them is put in a knapsack of size 1 \cite{HanKMY19}; this model also assumes that the algorithm is informed when the input ends.  Reservation with removability has been considered in \cite{BGLMR23}.  

Other proposed approaches for online knapsack include the possibility of {\em augmentation} of the knapsack size~\cite{IwamaZ10}, the possibility of items being split and partially included in the knapsack~\cite{HanM10}, and the study of the \emph{advice complexity} of the  problem~\cite{BKKR14,BFR24}. The work of Zhou \emph{et al.} on the weighted knapsack problem also makes (restrictive) assumptions about both the size of objects and their weight-to-size ratio \cite{ZCL08}.

For more information about the online knapsack problem, the interested reader may refer to a recent survey by B{\"{o}}ckenhauer {\it et al.} \cite{BHKRS26}.

\section{Algorithms with  Recourse} \label{sec:recourse}

\subsection{Lower Bounds  on the Best Competitive Ratio}

We propose (and analyze the {\cra} of) Algorithm \ref{algorec4} where for some integer $k \ge 3$ an upper bound on the item sizes of $u=1/k$ is assumed and at most $\lfloor k/2 \rfloor$ recourse operations are allowed. 
Small and large items have size in $(0, \frac{1}{k+2}]$ and $(\frac{1}{k+2},\frac{1}{k}]$, respectively. 
It follows that any feasible solution contains at most $k+1$ large items. Upon the arrival of item $o_i$, the high-level strategy of Algorithm \ref{algorec4} is to maintain a solution $S$ according to the following rules: 
$(i)$ do not modify $S$ if $v(S) \ge \frac{k+1}{k+2}$;  $(ii)$ when $v(S) < \frac{k+1}{k+2}$, attempt to exploit $o_i$ together with at most $\lfloor \frac{k}{2} \rfloor$ items from the buffer (via recourse operations) in order to construct a feasible solution of value at least $\frac{k+1}{k+2}$; $(iii)$
 if step $(ii)$ fails, retain in the current solution the $\lceil \frac{k}{2}\rceil$ smallest large items,  
the $\lfloor \frac{k}{2}\rfloor$ largest large items, and all small items revealed so far.


\begin{algorithm}[!t]
\small
\caption{\label{algorec4} at most $\lfloor \frac{k}{2}\rfloor$ recourse operations, $u \le 1/k$,   for any integer $k\ge 3$}
\begin{algorithmic}[1]
\State $S \gets \emptyset$
\Comment \emph{Small and large items have size in $(0,\frac{1}{k+2}]$ and  $(\frac{1}{k+2},\frac{1}{k}]$, respectively}

\While {a new item $o_i$ arrives}

    \If {$v(S) \ge \frac{k+1}{k+2}$}

    \State \label{lineraf2_gen} Put $o_i$ in the buffer

     \Else \If{$v(S \cup \{o_i\})  \leq 1$}
    
          \State \label{lineincl2_gen} $S \gets S \cup \{o_i\}$
          
          \Else
          

          \If{\label{lineL2_gen} $v(\{o \in S: |o| > \frac{1}{k+2} \} \cup \{o_i\} ) \leq 1$}
              \State  $S \gets S \cup \{o_i\}$

              \While {$v(S)>1$}
                  \State \label{linetrim3_gen} Remove a small item from $S$ \Comment \emph{chosen arbitrarily} 
               \EndWhile
          \Else

           \State \label{linecalT_gen}  ${\cal T} \gets$ collection of all  sets of at most $k+1$ elements of $\{o_1, \ldots, o_i\}$ 
          
          \Statex \hspace{26.5mm} which contain at most $\lfloor \frac{k}{2}\rfloor$ buffered items

           \If{$\exists T \in {\cal T}$ such that $1 \ge v(T) \ge \frac{k+1}{k+2}$}
                
                \State \label{linetef2_gen} $S \gets T$ 

            \Else  
             \State \label{line2396_gen} Keep in $S$ the $\lceil \frac{k}{2}\rceil$ smallest large items seen so far,  the $\lfloor \frac{k}{2}\rfloor$ largest 
             
             \Statex \hspace{23.5mm} large items seen so far, and every small item seen so far

\EndIf

\EndIf

\EndIf

\EndIf
\EndWhile

\Statex
\end{algorithmic}
  \vspace{-0.4cm}%
\end{algorithm}

\begin{theorem} \label{thmk1k2} For an integer $k\ge 3$, Algorithm \ref{algorec4} is $\frac{k+1}{k+2}$-competitive when $u \le 1/k$.  In addition, Algorithm \ref{algorec4} uses at most $\lfloor \frac{k}{2}\rfloor$ recourse operations.
\end{theorem}
\pjdelete{(Sketch)}
\begin{proof} 
The proof is by induction until it possibly happens that a feasible solution of value at least $\frac{k+1}{k+2}$ is built. As soon as 
the current solution has value at least $\frac{k+1}{k+2}$ (which occurs at line \ref{lineraf2_gen}, or when the while loop at line \ref{linetrim3_gen} is executed,\footnote{The while loop always creates a feasible solution of value at least 
$1-\frac{1}{k+2}=\frac{k+1}{k+2}$.} or at line \ref{linetef2_gen}), the algorithm 
 is clearly $\frac{k+1}{k+2}$-competitive  because the optimal value is at most 1. 

\pamargincomment{Made minor changes to save a line: Changed "Let us consider" to "Consider" and "We first" to "First,"}Consider the cases other than those of lines \ref{lineraf2_gen},  \ref{linetrim3_gen}, and \ref{linetef2_gen}. 
First, observe that if $v(S) < \frac{k+1}{k+2}$, then every small item disclosed so far belongs to $S$, and every item placed in the buffer is large. Indeed, an item is put into $S$ if it fits (cf. line \ref{lineincl2_gen}), which is always the case for small items because their size is at most $\frac{1}{k+2}$.

Let $OPT_{i}$  (resp., $S_i$) be an optimal solution (resp., the solution produced by Algorithm \ref{algorec4}) for the item set $\{o_1, \ldots,o_i\}$. Let $S_i^s$ and $S_i^\ell$ denote the restriction of $S_i$ to small and large items of $\{o_1, \ldots,o_i\}$, respectively. Let $Q_i$ be a subset of large items taken in $\{o_1, \ldots,o_i\}$ whose value is maximum under the constraint $v(Q_i) \le 1$. We claim that Algorithm \ref{algorec4} maintains a feasible solution which guarantees 
\begin{equation} \label{recinvariant_gen}
v(S_{j}^\ell) \ge \frac{k+1}{k+2} v(Q_{j})  
\end{equation}
for $j=1..i-1$ until it possibly happens that $v(S_i) \ge \frac{k+1}{k+2}$.  
The base case ($j=1$) of the induction hypothesis $(\ref{recinvariant_gen})$ holds trivially: if $o_1$ is large then $S_1^\ell=Q_1=\{o_1\}$, otherwise $S_1^\ell=Q_1=\emptyset$, so $v(S_1^\ell)=v(Q_1)$.  
From now on, we assume that $(\ref{recinvariant_gen})$ holds for $j=i-1$, right before $o_i$ is disclosed. 
Our proof also relies on the following observation  whose proof is \pjmargincomment{Kept the observation proof in the appendix so it does not interfere with reading the proof - it's still small so it can be added.}\pjreplace{deferred to the full version}{given in Appendix~\ref{app:obs1}}. 
\begin{observation} \label{intermobs} $v(S_j^\ell) \ge \frac{k+1}{k+2} v(Q_j)$  implies 
$v(S_j) \ge \frac{k+1}{k+2} v(OPT_{j})$.
\end{observation}

At line \ref{lineincl2_gen} of Algorithm \ref{algorec4}, $o_i$ is put into the current solution because it fits. If $o_i$ is small, then $S_i^\ell=S_{i-1}^\ell$, and $Q_i=Q_{i-1}$, so $v(S_i^\ell) = v(S_{i-1}^\ell)  \ge \frac{k+1}{k+2} v(Q_{i-1})= \frac{k+1}{k+2} v(Q_{i})$.  If $o_i$ is large, then $S_i^\ell=S_{i-1}^\ell \cup \{o_i\}$, and $v(Q_i) \le v(Q_{i-1}) + |o_i|$, so  
$v(S_i^\ell) = v(S_{i-1}^\ell) + |o_i| \ge \frac{k+1}{k+2} v(Q_{i-1}) + |o_i| \ge \frac{k+1}{k+2} v(Q_{i})$. In both cases, the insertion  of $o_i$ maintains the property of being $\frac{k+1}{k+2}$-competitive.

The last case to be analyzed is line \ref{line2396_gen}. In this context we have $v(S_{i-1})<\frac{k+1}{k+2}$, $v(S_{i-1} \cup \{o_i\})>1$ and $v(S^\ell_{i-1} \cup \{o_i\})>1$, implying that $o_i$ is a large item.

Suppose that at some point $i$, $\{o_1,\ldots, o_i\}$ contains for the first time a set $Y$ of $k+1$ large items such that $v(Y)\le 1$. This means that $o_i \in Y$. Since Algorithm \ref{algorec4} maintains the $\lceil \frac{k}{2}\rceil$ smallest large items seen so far, it can exploit $o_i$ and the $\lfloor \frac{k}{2}\rfloor$ recourse operations to produce a set $Y'$ of $k+1$ large items such that $v(Y') \le v(Y) \le 1$. Since every large item has size at least $\frac{1}{k+2}$, it holds that  $v(Y') \ge \frac{k+1}{k+2}$. Thus, the algorithm is able to produce a feasible solution of value at least  $\frac{k+1}{k+2}$ as soon as 
a feasible set of $k+1$ large items exists. 

It remains to consider the case where every set $Y$ of large items such that $v(Y)\le 1$ satisfies $|Y| \le k$.  This is true, in particular, for $S_i^\ell$. Moreover,  $S_i^\ell$ cannot contain less than $k$ large items because the test of line \ref{lineL2_gen} is false. Indeed, any set of $k$ large items has size at most 1 because every item has size at most $1/k$. One cannot take a large item of the buffer, say $\hat o \notin S_i^\ell$, and have  $v(S^\ell_{i} \cup \{\hat o\}) \le 1$, because every set $Y$ of large items such that $v(Y)\le 1$ satisfies $|Y| \le k$. This means that $v(S^\ell_{i} \cup \{\hat o\}) > 1 \pjreplace{\Leftrightarrow}{\Rightarrow} v(S^\ell_{i}) > 1-\frac{1}{k}=\frac{k-1}{k}$. 
If $v(Q_i) \le \frac{k+1}{k+2}$, then $\frac{v(S^\ell_{i})}{v(Q_i)} > \frac{k-1}{k} \cdot \frac{k+2}{k+1} \ge \frac{k+1}{k+2}$ when $k \ge 3$. 

The final case is when $v(Q_i) > \frac{k+1}{k+2}$. By the definition of $Q_i$, $v(Q_i)$ is non-decreasing with $i$. Let us suppose that $i$ is the first moment that $v(Q_i) > \frac{k+1}{k+2}$. This means that $Q_i$ contains $o_i$. 
Let $\{o'_1, \ldots,o'_{\lceil \frac{k}{2}\rceil}\} \cup \{p_1,  \ldots,p_{\lfloor \frac{k}{2}\rfloor}\}$ denote the $k$ (large) items of $Q_{i-1}$ where we assume w.l.o.g. that 
$|o'_1| \le |o'_2|  \le \ldots \le |o'_{\lceil \frac{k}{2}\rceil}| \le |p_1| \le |p_2| \le \ldots \le |p_{\lfloor \frac{k}{2}\rfloor}|$. 
We have $Q_i=Q_{i-1} \Delta \{o'_1,o_i\}$. 

Since $S_{i-1}^\ell$ keeps the $\lfloor \frac{k}{2}\rfloor$ largest items seen so far, we can assume that 
$\{p_1,  \ldots,p_{\lfloor \frac{k}{2}\rfloor}\} \subseteq S_{i-1}^\ell$. 
In other words, at most $\lfloor \frac{k}{2}\rfloor$ recourse operations are needed to produce $Q_i$ from $S_{i-1}^\ell$ 
because $\{p_1,  \ldots,p_{\lfloor \frac{k}{2}\rfloor}\}$ and $o_i$ are accessible without any recourse. Thus, $Q_i \in {\cal T}$ at line \ref{linecalT_gen} of Algorithm \ref{algorec4}, meaning that the algorithm can produce a solution of value at least  $\frac{k+1}{k+2}$.


Finally,  the following Observation~\ref{obs_correctness}, shows that Algorithm \ref{algorec4} is correct:  it produces a feasible solution and uses at most 
$\lfloor k/2 \rfloor$ recourse operations . \qed

\pjedit{

\begin{observation} \label{obs_correctness} Algorithm \ref{algorec4} is correct.  
\end{observation}

\begin{proof}


Let us first consider how many times the recourse is used. By construction, recourse is made for possibly multiple items but during a single round, and for each recourse a solution whose value is between $\frac{k+1}{k+2}$ and $1$ is produced.  

At line \ref{linecalT_gen}, at most $\lfloor k/2\rfloor$ items are taken from the buffer. 
No recourse is necessary for building $S_i$ from $S_{i-1}$ and $o_i$ at line \ref{line2396_gen}. Indeed, we have seen in the proof of Theorem \ref{thmk1k2} that $S_{i-1}$ contains exactly $k$ large items. More precisely,  
$S_{i-1}$ contains the 
$\lceil \frac{k}{2}\rceil$ smallest large items, say $\{a_1, \ldots,a_{\lceil \frac{k}{2}\rceil}\}$, and the 
$\lfloor \frac{k}{2}\rfloor$ largest large items, say $\{b_1, \ldots,b_{\lfloor \frac{k}{2}\rfloor}\}$, of $\{o_1, \ldots,o_{i-1}\}$, 
where $|a_1| \le \ldots \le |a_{\lceil \frac{k}{2}\rceil}| \le |b_1| \le \ldots \le |b_{\lfloor \frac{k}{2}\rfloor}|$. If $|o_i| < |a_{\lceil \frac{k}{2}\rceil}|$, then $S_{i}=S_{i-1} \Delta \{o_i,a_{\lceil \frac{k}{2}\rceil}\}$. If $|o_i| > |b_{1}|$, then $S_{i}=S_{i-1} \Delta \{o_i,b_1\}$. If $|a_{\lceil \frac{k}{2}\rceil}| \le |o_i| \le |b_{1}|$, then $S_{i}=S_{i-1}$. 

It remains to prove that the algorithm maintains a solution whose value does not exceed 1. It is immediate at lines \ref{lineincl2_gen}, \ref{linetrim3_gen} and \ref{linetef2_gen}. Let us show that $v(S_i) \le 1$ holds at line \ref{line2396_gen}. We have $v(S_i^\ell) < \frac{k+1}{k+2}$ because $S_i^\ell \in {\cal T}$. We also know that $v(S_i^s)\le \frac{1}{k}-\frac{1}{k+2}$ because $v(S_{i}^s)=v(S_{i-1}^s) = v(S_{i-1})-v(S_{i-1}^\ell) < \frac{k+1}{k+2} - v(S_{i-1}^\ell)$, and  
$v(S_{i-1}^\ell)>1-|o_i|\ge 1-1/k$.  Thus, it holds that $v(S_i) = v(S_i^s) + v(S_i^\ell) \le \frac{k+1}{k+2}+\frac{1}{k}-\frac{1}{k+2} \le 1$ when $k \ge 2$. \qed
\end{proof}}

\end{proof}

Extending Theorem \ref{thmk1k2} to $k=2$ yields a {\cra} of $3/4$ using at most one recourse operation when $u\le 1/2$. In fact, a modified version of Algorithm \ref{algorec4}, \pjedit{Algorithm~\ref{algorec5},} using four classes of items   (small, medium, large and extra-large)  
gives the following stronger result. 
\begin{theorem} \label{thmk34} There exists a deterministic $\frac{3}{4}$-competitive algorithm using at most one recourse operation when $u \le 5/8$.
\end{theorem}

\pjmargincomment{Moved this from the appendix.}
\pjedit{

\begin{algorithm}
\caption{\label{algorec5} at most 1  recourse operation and $u \le 5/8$}
\begin{algorithmic}[1] 
\State   $S \gets \emptyset$,  
\Comment{\emph{Small, medium, large and extra-large items have size in $(0,\frac{1}{4}]$, $(\frac{1}{4},\frac{3}{8}]$, $(\frac{3}{8},\frac{1}{2})$, and $[\frac{1}{2},\frac{5}{8}]$, respectively}}

\While{a new item $o_i$ arrives}
    \If{$v(S) \ge \frac{3}{4}$}
         \State   \label{lineraf58} Put $o_i$ in the buffer
    \Else
    
    \If{$v(S \cup \{o_i\})  \leq 1$}
    
    \State \label{lineincl58} $S \gets S \cup \{o_i\}$
 
 \Else

\State \label{lineL_58} Let $S^\ell=\{o \in S: |o| > \frac{1}{4} \}$ 

\If{\label{lineL2_58} $v(S^\ell \cup \{o_i\} ) \leq 1$}
\State  $S \gets S \cup \{o_i\}$

\While{$v(S)>1$}
\State  \label{linetrim_58} Remove a small item from $S$ (chosen arbitrarily)

\EndWhile

\Else

\State \label{linecalT_58} ${\cal T} \gets$ collection of all sets of at most $3$ elements of $\{o_1, \ldots, o_i\}$ 

\Statex \hspace{28.2mm} which contain at most $1$ item of the buffer

\If{\label{lineExists_58} $\exists T \in {\cal T}$ such that $1 \ge v(T) \ge \frac{3}{4}$}

\State   \label{lineSet_58} $S \gets T$ 

\Else 

\State    \label{lineTamb_58} If $S^\ell \cup \{o_i\}$ contains two extra-large items, then place the

\Statex \hspace{25mm} largest one into the buffer, and keep the smallest one.

\Statex \hspace{25mm} Otherwise, if $S^\ell \cup \{o_i\}$ contains an extra large item and a

\Statex \hspace{25mm} large item, then place the large item into the buffer and

\Statex \hspace{25mm} keep the extra-large item. 
Otherwise, if $S^\ell \cup \{o_i\}$ contains 

\Statex \hspace{25mm} 3 items (i.e., 3 medium items, or 2 medium items and 1

\Statex \hspace{25mm} large item), then place the median one into the buffer, and 

\Statex \hspace{25mm} keep the remaining 2 items (i.e., the smallest and the 

\Statex \hspace{25mm} largest)  

\EndIf

\EndIf

\EndIf 
\EndIf 

\EndWhile

\end{algorithmic}
\end{algorithm}

\pjreplace{We}{For the proof, we} are going to analyze Algorithm \ref{algorec5}  
which tries to produce a solution whose value is between $3/4$ and $1$. Such a solution is called a \emph{target} solution. We are particularly interested in \emph{non-small items}, i.e., items whose size is strictly larger than $1/4$. 
More specifically, we say that an item $o$ is small, medium, large, and extra-large when $|o| \in (0,\frac{1}{4}] $, 
$|o| \in (\frac{1}{4}, \frac{3}{8}]$, $|o| \in (\frac{3}{8}, \frac{1}{2})$, and $|o| \in [\frac{1}{2},\frac{5}{8}]$, respectively.  Hence, non-small items are either medium, large or extra-large.   At line \ref{lineTamb_58}, we use the notion of \emph{median} item. Given three items $\{o_a,o_b,o_c\}$, we say that $o_b$ is a median item if $|o_a| \le |o_b| \le |o_c|$. 

\medskip

\pjdelete{\noindent {\bf Theorem \ref{thmk34}.} \emph{ There exists a deterministic $\frac{3}{4}$-competitive algorithm using at most one recourse operation when $u \le 5/8$.}}

\begin{proof} Algorithm \ref{algorec5} follows the same line as Algorithm \ref{algorec4}. The same goes for their analyses. Indeed, using the same notations and arguments as in the proof of Theorem \ref{thmk1k2} it is not difficult to prove that Algorithm \ref{algorec5} is $3/4$-competitive, except when line \ref{lineTamb_58} is executed.  This case happens when 
\begin{equation} \label{aretenir581}
v(S^\ell_{i-1} \cup \{o_i\}) > 1
\end{equation}
and 
\begin{equation} \label{aretenir582}
v(S^\ell_{i-1}) \le v(S_{i-1}) < 3/4
\end{equation}
hold, where $S^\ell_{i-1}$ designates $S^\ell$ right before $o_i$ is disclosed. By construction, $S_{i-1}$ contains all the small items disclosed so far.  It follows from (\ref{aretenir581}) that $o_i$ is non-small. 


Since $u<3/4$ , and because non-small items have size strictly larger than $1/4$, a combination of non-small items whose value lies in the interval $[3/4,1]$ must contain two or three elements. Table \ref{tabcombin2} provides the list of combinations of two or three non-small items, and the possibility that they satisfy the lower and upper bounds for being a target solution. 

\begin{table}
\caption{\label{tabcombin2} Eight combinations of non-small items and comments on their total value. For example, ``always $\ge 3/4$'' means that the total value of the corresponding combination is always above $3/4$, whereas ``possibly $>1$'' indicates that the total value may exceed 1.}
\begin{center}\begin{tabular}{|c|l|c|c|}
\hline
\#1&2 medium items&possibly $< 3/4$&always $\le 1$\\
\#2&1 medium item + 1 large item&possibly $< 3/4$&always $\le 1$\\
\#3&1 medium item + 1 extra large item&always $\ge 3/4$&always $\le 1$\\
\#4&2 large items&always $\ge 3/4$&always $\le 1$\\
\#5&3 medium items&always $\ge 3/4$&possibly $>1$\\
\#6&1 large item + 1 extra large item&always $\ge 3/4$&possibly $>1$\\
\#7&2 extra large items&always $\ge 3/4$&possibly $>1$\\
\#8&2 medium items + 1 large item&always $\ge 3/4$&possibly $>1$\\
\hline
\end{tabular}
\end{center}
\end{table}

For every other combination of non-small items it is guaranteed that its value is out of the interval $[3/4,1]$.   
Combinations {\#3} and {\#4} are particularly interesting because they are guaranteed to be target solutions.

Note that $|S_{i-1}^\ell| \in \{1,2\}$ because $|S_{i-1}^\ell| =0$ and $|S_{i-1}^\ell| \ge 3$ would contradict  (\ref{aretenir581}) and (\ref{aretenir582}), respectively.

Having just a medium item in $S_{i-1}^\ell$ would contradict (\ref{aretenir581}). Inequality (\ref{aretenir582}) would be contradicted if $S_{i-1}^\ell$ contained one of the following combinations of two non-small items: 1 extra-large item plus another non-small item, or 2 large items. Thus, the only possibilities for the composition of $S_{i-1}^\ell$ are: 1 large item, 1 extra-large item, 2 medium items, and 1 medium plus 1 large items. 
All the possible cases are listed in Table \ref{atct}.

\begin{table}
\caption{\label{atct}All the possible cases at line \ref{lineTamb_58}.}
\begin{center}
\begin{tabular}{|c|c|c|c|}
\hline
composition of $S_{i-1}^\ell$&$o_i$ is extra-large&$o_i$ is large& $o_i$ is medium\\
\hline
1 extra-large item&case $1$&case $2$&case $3$\\
1 large item&case $4$&case $5$&case $6$\\
2 medium items&case $7$&case $8$&case $9$\\
1 medium and 1 large items&case $10$&case $11$&case $12$\\
\hline
\end{tabular}
\end{center}
\end{table}

For cases 3, 5, 7, 10 and 11, a target solution similar to combinations \#3 or \#4 can be created, so they are incompatible with the execution of line \ref{lineTamb_58}.  Case 6 does not occur because it contradicts (\ref{aretenir581}). Thus, line \ref{lineTamb_58} covers all the remaining cases (1,2, 4, 8, 9, and 12), meaning that Algorithm \ref{algorec5} is correct.

Let $Q_j$ be a set of non-small items of largest value within $\{o_1, \ldots,o_j\}$, without exceeding 1, for $j=1..i$. As shown in the proof of Theorem \ref{thmk1k2}, it suffices to prove $v(S_i^\ell) \ge \frac{3}{4} v(Q_i)$ to get that Algorithm \ref{algorec5} is $\frac{3}{4}$-competitive at step $i$. And as in the proof of Theorem \ref{thmk1k2}, we make the induction hypothesis that $v(S_{i-1}^\ell) \ge \frac{3}{4} v(Q_{i-1})$.

Consider the cases 1, 2 and 4. By construction, the algorithm maintains one extra-large item in $S^\ell_i$, so $v(S^\ell_i) \ge 1/2$. If $Q_i$ consists of a single non-small item, then $v(Q_i) \le 5/8$. We deduce that  $\frac{v(S_i^\ell)}{v(Q_i)} \ge \frac{1}{2} \cdot \frac{8}{5}=\frac{4}{5} > \frac{3}{4}$. Suppose $Q_i$ consists of two non-small items. If there is a medium item in $Q_i$, then it is possible to produce a combination \#3 with a single recourse (use the extra-large item of $S_i^\ell$), and reach the value of at least $3/4$: contradiction with the fact that no such combination exists in ${\cal T}$. Every pair of non-small items which does not include any medium item has value at least $3/4$, so we again reach a contradiction. Finally, suppose that $Q_i$ consists of three non-small items. Then, $Q_i$ must include a medium item, implying that it is possible to produce a combination \#3 with a single recourse: contradiction.

Consider the cases 8, 9 and 12.  By construction, the algorithm maintains two non-small items in $S^\ell_i$, the smallest and the largest, 
$S^\ell_i$ is either a pair of medium items or a medium and a large items, and $v(S^\ell_i) > 2\cdot\frac{1}{4}=\frac{1}{2}$. If $Q_i$ consists of a single non-small item, then $v(Q_i) \le 5/8$, and we deduce that  $\frac{v(S_i^\ell)}{v(Q_i)} \ge \frac{1}{2} \cdot \frac{8}{5} > \frac{3}{4}$.  

Suppose $Q_i$ consists of two non-small items. If there is an extra-large item in $Q_i$, then it is possible to produce a combination \#3 with a single recourse (use a medium item of $S_i^\ell$), and reach the value of at least $3/4$: contradiction with the fact that no such combination exists in ${\cal T}$. 

Suppose $Q_i$ consists of two medium items. \pjreplace{Because $v(Q_i)\geq v(S_i^\ell)$}{Since $Q_i$ has maximum value and a medium and a large item always fit together}, $S_i^\ell$ also contains two medium items.  
Since, by construction, the algorithm keeps the smallest and the largest items, we have $v(S^\ell_i)=x+y$ where $\frac{1}{4} < x \le y \le \frac{3}{8}$, and $v(Q_i)\le 2y$, so $\frac{v(S^\ell_i)}{v(Q_i)} \ge \frac{1/4+3/8}{2\cdot 3/8}=\frac{5}{6} >\frac{3}{4}$.  

Suppose $Q_i$ consists of a medium and a large items. Since, by construction, the algorithm keeps the smallest and the largest non-small items, $S_i^\ell$ also contains a medium and a large items. Moreover, the large item of $S_i^\ell$ is at least as large as the large item of $Q_i$. We have $v(S^\ell_i)=x+y$ where $\frac{1}{4} < x  \le \frac{3}{8} < y \le \frac{1}{2}$, and $v(Q_i)\le \frac{3}{8} +y$, so $\frac{v(S^\ell_i)}{v(Q_i)} \ge  \frac{1/4+y}{3/8+y} \ge  \frac{1/4+3/8}{2\cdot 3/8}=\frac{5}{6} >\frac{3}{4}$.  

Suppose $Q_i$ consists of two large items, corresponding to combination \#4 whose value is always in $[3/4,1]$. Such a target solution belongs to ${\cal T}$ because at most one recourse is needed to produce it: contradiction.  

The remaining case is when $Q_i$ consists of three items, say $\{o_k,o_j,o_i\}$ where $k<j<i$. 
The set $Q_i$ satisfies $v(Q_i) \ge 3/4$, and it must contain three medium items, or two medium items plus one large item, otherwise $v(Q_i)>1$. 
We claim that a contradiction is reached because a single recourse is enough for constructing a solution $T \in {\cal T}$ such that $v(T) \in [\frac{3}{4},1]$ with the use of $o_i$ and the non-small items of $S_{i-1}^\ell$.  


Indeed, by construction, the algorithm keeps the smallest non-small item, so we have 
$S^\ell_{i-1}=\{o_a,o_b\}$ where $o_i \notin \{o_a,o_b\}$, $|o_a| \le |o_b|$, 
and $|o_a| \le \min \{|o_k| , |o_j|\}$. Suppose without loss of generality that $o_a \neq o_j$ (otherwise, swap $o_j$ and $o_k$ in our reasoning). One can use (if necessary) the recourse to get $o_j$ and produce the set $\{o_a,o_i,o_j\}$ such that $v(\{o_a,o_i,o_j\}) \le v(\{o_k,o_i,o_j\}) \le 1$.   Moreover, $v(\{o_a,o_i,o_j\}) \ge 3/4$ because every member of $\{o_a,o_i,o_j\}$ is non-small. \qed
\end{proof}

}

Compared with \cite{BKMRSW26}, for identical competitive ratios, our results use fewer recourse operations, albeit under the assumption of bounded item sizes. Likewise, compared with \cite{GourvesP24}, for identical bounds on item sizes, we achieve better competitive ratios thanks to the use of recourse.

\subsection{Upper Bounds on the Best Competitive Ratio}

It is known from  \cite[Thm. 1]{BKMRSW26} that deterministic algorithms that use at most $r$ recourse operations have a {\cra} of at most $\frac{r+1}{r+2}$. This bound is established via 
instances in which the largest item has size $\frac{r+1}{r+2}+\varepsilon$. Consequently, both the size of the largest item and the upper bound on the {\cra} increase with $r$. However, it is also known that the optimal {\cra} increases as $u$ decreases \cite{GourvesP24}. This motivates the need for novel upper bounds on the {\cra} that reflect this dependence on 
$u$.

\begin{lemma} \label{lemmaub1rec2}
For every integer $k \ge 3$, there exists an upper bound of $\frac{2k-1}{2k}$ on the  competitive ratio of deterministic algorithms with at most $k-2$ recourse operations when  $\frac{1}{k}< u \le 1$.
\end{lemma}
\begin{proof} Take any online algorithm and suppose that an adversary reveals the following instance:  $k-1$ \emph{large} items of size $\frac{1}{k}+\varepsilon$ followed by  $2k$ \emph{small} items of size $\frac{1}{2k}$, where $ \frac{1}{2k} \gg \varepsilon >0$.  
The instance can be interrupted at any moment if the algorithm's solution is (asymptotically) $(1-\frac{1}{2k})$-approximate, or worse. 

The first $k$ items fit into the knapsack. If the algorithm discards any of them, then the {\cra} is at most $1-\frac{1}{2k-1}$ as $\varepsilon \rightarrow 0$. Indeed, the optimal solution has value $(k-1)(\frac{1}{k}+\varepsilon)+\frac{1}{2k}$ whereas the algorithm's solution has value at most  $(k-1)(\frac{1}{k}+\varepsilon)$. Henceforth, we assume that the first $k$ items belong to the solution.

When the $(k+1)$-th item is revealed, it does not fit into the current solution, while the optimal value remains $(k-1)(\frac{1}{k}+\varepsilon)+\frac{1}{2k} = 1-\frac{1}{2k} +(k-1)\varepsilon$. The optimal value increases to 1 only upon the arrival of the final item, when the optimum consists of 
$2k$ small items. 
At this point, the algorithm may 
$(i)$ replace 1 large item with 1 small item, or
$(ii)$ buffer a small item,  or $(iii)$  replace 1 large item with 2 small items. 


Action $(i)$ decreases the value of the current solution by $\frac{1}{2k}+\varepsilon$, yielding a {\cra} of $1-\frac{1}{2k-1}$; hence, the adversary can terminate the instance. Actions $(ii)$ and $(iii)$  do not significantly affect the solution value, but action $(iii)$ can be performed at most  $k-2$ times, as each use requires a recourse operation.

Since the instance contains $k-1$ large items, the algorithm must either perform action $(i)$, or reach the end of the instance, where the optimal value is 1 
while the algorithm's solution has value at most $(k-1)(\frac{1}{k}+\varepsilon)+\frac{1}{2k}= 1-\frac{1}{2k} +(k-1)\varepsilon$. In any case, the {\cra} is at most $1-\frac{1}{2k}$ as $\varepsilon \rightarrow 0$.  \qed
\end{proof}

In the same spirit as Lemma \ref{lemmaub1rec2}, the following upper bounds can be derived.

\begin{lemma} \label{lemmaub1rec3} For every integer $k \ge 2$, there exists an upper bound of $\frac{2k}{2k+1}$ on the {\cra}  of deterministic algorithms with at most $k-1$ recourse operations when $\frac{2}{2k+1} < u \le 1$.
\end{lemma}

\pjmargincomment{Moved from appendix.}
\pjedit{
\begin{proof} Take any online algorithm and suppose that an adversary discloses the following instance:  $k$ \emph{large} items of size $\frac{2}{2k+1}+\varepsilon$ followed by  $2k+1$ \emph{small} items of size $\frac{1}{2k+1}$. We suppose that $ \frac{1}{2k+1} \gg \varepsilon >0$. The size of the largest item is $\frac{2}{2k+1}+\varepsilon$.
The instance can be interrupted at any moment by the adversary when the solution maintained by the algorithm is (asymptotically) $(1-\frac{1}{2k+1})$-approximate, or worse. 

The first $k$ items (which are identical) fit into the knapsack. If one of these items is not kept by the algorithm, then the competitive ratio is at most $1-\frac{1}{k}$ when $\varepsilon$ tends to 0. From now on we suppose that the $k$ first items of the instance belong to the solution. 

When the $(k+1)$-th item of the instance is disclosed, it does not fit into the solution. However, the value of an optimal solution \pjreplace{($k*(\frac{2}{2k+1}+\varepsilon$)}{$k(\frac{2}{2k+1}+\varepsilon)$} does not change. The optimal value changes only if the very last item is disclosed (i.e., it becomes 1 with $2k+1$ small items).

In between, the algorithm can 
$(1)$ replace one large item with one small item, or
$(2)$ buffer a small item,  or $(3)$  replace one large item with two small items. 

Note that the first action decreases the value of the current solution by $\frac{1}{2k+1}+\varepsilon$, giving a competitive ratio of $1-\frac{1}{2k}$ when $\varepsilon$ tends to zero. Therefore, the adversary can stop the instance if the first action is done. The second and third actions do not significantly change the value of the current solution (no decrease or a decrease by $\varepsilon$). However, the third action can be performed at most  $k-1$ times because each of them requires a recourse. 

Since the instance comprises $k$ large items, the algorithm is forced, either to do the first action, or to reach the very end of the instance where the value of the optimum becomes 1 and the algorithm's solution has value at most $k(\frac{2}{2k+1}+\varepsilon)= 1-\frac{1}{2k+1} +k\varepsilon$. In any case, 
the competitive ratio is at best $1-\frac{1}{2k+1}$ when $\varepsilon$ goes to zero. \qed
\end{proof}
}

We now present a complementary upper bound on the {\cra} for the case of a single recourse operation.

\begin{lemma} \label{lemmaub1rec1} The competitive ratio of any algorithm with at most one recourse operation is upper-bounded by $\sqrt{\frac{9-\sqrt{17}}{8}} \approx 0.78$ when $u\ge \frac{9-\sqrt{17}}{8} \approx 0.61$.
\end{lemma}

\pjmargincomment{Moved from appendix}
\pjedit{
\begin{proof} Let $s=\frac{9-\sqrt{17}}{8} \approx 0.61$.  
The first items $o_1$ and $o_2$ have size $1-s+\varepsilon$ and $s$, respectively.\footnote{We suppose that $\varepsilon>0$ and $s \gg \varepsilon$ hold.} If  $o_2$ does not replace $o_1$, then the competitive ratio is  $\frac{1-s}{s}$ when $\varepsilon$ goes to $0$. From now on, suppose that $\{o_2\}$ is the current solution and $o_1$ is in the buffer. 
The third item $o_3$ has size $1-s+\varepsilon$.  
If the solution remains $\{o_2\}$, then we get a ratio of $\frac{s}{2(1-s+\varepsilon)}$ which tends to $\frac{s}{2(1-s)}=\sqrt{s}$ when $\varepsilon$ goes to $0$. Otherwise, one can use the recourse to produce the solution $\{o_1,o_3\}$. Then, a fourth item $o_4$ of size $1-s$ is disclosed. The optimum $\{o_2,o_4\}$ has size one, and the size of the current solution is at most $2(1-s+\varepsilon)$, so the competitive ratio \pjreplace{is $2(1-s)$ which tends to $\sqrt{s}$}{tends to $2(1-s)=\sqrt{s}$} when $\varepsilon$ goes to $0$.   
Given that $s>1-s+\varepsilon$ and $\frac{1-s}{s} < \sqrt{s}$, we get an upper bound of $\sqrt{s} \approx 0.78$ when $u \ge s$. \qed
\end{proof}
}

Let us summarize the bounds on the {\cra} when a \emph{single} recourse operation is possible (see Figure \ref{fig4}). There exist lower bounds of ${4}/{5}$ when $u \le 1/3$ (Theorem \ref{thmk1k2}), ${3}/{4}$ when $u \le  5/8$ (Theorem \ref{thmk34}), and $2/3$ when $u \le 1$ (\cite[Thm. 2]{BKMRSW26}). There exist upper bounds of ${5}/{6}$ when $u >1/3$ (Lemma \ref{lemmaub1rec2}, $k=3$), 
${4}/{5}$ when $u >2/5$ (Lemma \ref{lemmaub1rec3}, $k=2$), 
$((9-\sqrt{17})/8)^{1/2}$ when $u \ge (9-\sqrt{17})/8$ (Lemma \ref{lemmaub1rec1}), 
and ${2}/{3}$ when $u >2/3$ (\cite[Thm. 1]{BKMRSW26}).

In this section, we have proposed upper and lower bounds on the {\cra} for certain combinations of upper bounds on object sizes, along with a limit on the number of allowable recourse operations. It would be interesting either to provide results for all possible combinations or to show that, for some combinations, recourse does not yield any significant improvement in the optimal {\cra}. For example, what would be the best {\cra} for all $u$ with at most one recourse? Our results provide answers for $u \geq 1/3$, and we believe that a single recourse can only marginally improve the situation when $u$ is very close to zero. Resolving this conjecture would be of great interest.

\section{Adaptive Algorithms} \label{sec:adaptive}


\subsection{An Adaptive Algorithm with Removal Only}

Algorithms that are aware of an upper bound on item sizes and whose {\cra} is lower bounded by $\frac{1}{\phi}$, but can exceed $\frac{1}{\phi}$ for most size upper bounds, exist \cite{GourvesP24}. Importantly, the smaller the size bound, the stronger the \cra\ guarantees they provide (almost tight per upper bound). This raises the question:

\begin{quote}
\centering\textit{Can we design algorithms whose competitive 
ratio (inversely) depends on the maximum item size, without any prior information?}
\end{quote}
Crucially, we want our adaptive algorithms' \cra\ to still be lower bounded by $\frac{1}{\phi}$, in the worst case. 

We begin by stating some useful definitions and theorems 
for the case of removable items without recourse. Consider an upper bound $u \in (0,1]$ on the size of every item. \pjedit{In this section, we explicitly consider $u$ to be the maximum item size in the instance (unknown to the algorithm).} Then, in~\cite{GourvesP24} the following is shown:

\begin{theorem}[\cite{GourvesP24}] \label{thm11} For all integers $k>0$, there exists a $\gamma_k$-competitive algorithm for a known  $u \in (0, \frac{1-\gamma_k}{\gamma_k^2}]$, where $\gamma_k$ is a constant defined as $\gamma_k=\frac{k-2+\sqrt{k^2+4}}{2k}$. In addition, for all integers $k \ge 2$, the same algorithm is  $\frac{1-u}{(k-1)u}$-competitive when $u \in (\frac{1-\gamma_k}{\gamma_k^2}, \frac{1}{k-1}]$.
\end{theorem}

For all positive integers $k$, $\gamma_k \in (0,1]$ increases with $k$, and $\gamma_1=\frac{\sqrt{5}-1}{2}=\frac{1}{\phi}$. We will refer to this algorithm of \cite{GourvesP24} as ``\nonadapt '' here. \pjedit{Other useful theorems and lemmata from \cite{GourvesP24} are included in Appendix~\ref{ap:DetailsFromGP}.}

\nonadapt\ can be adapted to get improved guarantees (compared to $\frac{1}{\phi}$ \cite{IT02}) without knowing the item size upper bound  $u$ in advance, at the cost of a small penalty in the \cra . Our adaptive version, Algorithm~\ref{algo_adaptive} (a.k.a. \adapt), follows \nonadapt\ with the addition of lines \ref{lineMaxitem}, \ref{lineKcurr} and \ref{lineKover}. These lines are used to dynamically update the value of $k$ whereas \nonadapt\ directly uses a fixed value, optimal for the given item size upper bound $u$.

The adaptive version updates an over-estimation of the upper bound (setting $k'= \max\{k_{curr}-2, 1\}$ in line~\ref{lineKover}) whenever a larger item arrives and continues running with the corresponding threshold. We show that \pjreplace{its behavior matches that}{it preserves the guarantee} of \nonadapt\ tuned from the start to the final over-estimation. Since the true bound is smaller, the corresponding approximation guarantee is preserved.

Informally, \nonadapt\ defines a threshold $1-\gamma_k$ and classifies items as small or non-small. When a new item does not fit into the knapsack and the guarantee is not met, items are removed based on the total sum of non-small item sizes. Specifically, when the sum of non-small item sizes is infeasible, it removes the largest item, otherwise it removes small items until it gets a feasible solution.


\begin{definition}[$k$-small items]
An item $o$ is said to be $k$-small if $|o|\leq 1-\gamma_k$.
\end{definition}

The challenge is to identify the case where \adapt\ might remove an important item that \nonadapt\ would not. In this case \adapt\ would delete its current largest item: For a current estimation of $k$, say $k'$, the $k'$-non-small items seen-so-far do not fit in the knapsack, but for a later estimation, say $k''$, some of these items could be considered $k''$-small by \nonadapt\ - allowing the rest to fit in the knapsack, not triggering the large item deletion. However, the over-estimation of the upper bound raises the non-small size threshold. Moreover, we target a lower \cra\ which is easier to satisfy. Thus, as non-small items accumulate, the desired \cra\ is achieved before they exceed the knapsack capacity.

\pjmargincomment{Also kept those in the appendix.}
\pjreplace{For the following, all missing proofs of this section are included in \pareplace{Appendix~\ref{ap:omittedAdaptiveRemoval}. Furthermore any properties used from \cite{GourvesP24} appear in Appendix~\ref{ap:DetailsFromGP}.}{the full version of the paper.}}{The proof of Observation~\ref{obs:optimal_k} is in Appendix~\ref{ap:omittedAdaptiveRemoval}.}

\begin{observation}
\label{obs:optimal_k}
    Consider that the item upper bound is $u$. Then the optimal parameter $k$ for \nonadapt\ to guarantee the best possible approximation is: 
    \[
    k^* =  \max\Big\{k\in \mathbb{N}\text{ } \big | \text{ } \frac{2}{k-1+\sqrt{(k-1)^2+4}}\geq u \Big\}
    \]
\end{observation}


\begin{theorem}
    \label{thm:adaptiveMain}
    For all positive integers $k$, Algorithm~\ref{algo_adaptive}, a.k.a. \adapt, is \pjreplace{$\gamma_{\hat{k}}$}{$\gamma_{k'}$}-competitive when $u\leq \frac{2}{k-1+\sqrt{(k-1)^2+4}}$ where $\pjreplace{\hat{k}}{k'}= \max\{\pjreplace{k^*}{k}-2, 1\}$. \pjedit{In particular, \adapt\ is $\gamma_{\hat{k}}$-competitive where  $\hat{k} = \max\{k^*-2, 1\}$.} 
    \pjmargincomment{Think this makes it a bit clearer.}
\end{theorem}

\begin{algorithm}
\caption{: \adapt
\label{algo_adaptive}}
\begin{algorithmic}[1] 
\State $S \gets \emptyset$, $MaxItem \gets 0$, $B \gets \emptyset$

\While{a new item $o_i$ arrives}
\State $MaxItem \gets \max\{MaxItem, |o_i|\}$\label{lineMaxitem}

\State $k_{curr}\gets \max\Big\{k\in \mathbb{N}\text{ } \big | \text{ } \frac{2}{k-1+\sqrt{(k-1)^2+4}}\geq MaxItem \Big\}$\label{lineKcurr} \LineComment{{\em Optimal $k$  considering $MaxItem$ as the upper bound}}

\State $k'\gets \max\{k_{curr}-2, 1\}$\label{lineKover}

    \If{ \label{adline1024} $v(S) \ge \gamma_{k'}$}
        \State \label{adline1025}  Reject $o_i$ \Comment{{\em $S$ is not changed afterwards}}

\Else
\State \label{adline1027} $S \gets S \cup \{o_i\}$
    
\If{\label{adline1030} $v(S) > 1$}

\State Let $B=\{o \in S : |o| > 1-\gamma_{k'} \}$
\label{algStep:smallRemoval} \Comment{Reset B}

    \If{$B$ contains a subset $\hat S$ 
such that $1 \ge v(\hat S) \ge \gamma_{k'}$}

\State \label{adline291} $S \gets \hat S$ \Comment{{\em $S$ is not changed afterwards}}

\ElsIf{  $v(B) > 1$}

\State \label{adline1044} Remove the largest item of $S$ 

\Else

\While{$v(S)>1$ \label{adline1050}}
\State \label{adline1051} Remove from $S$ one element of $S \setminus B$ \Comment{{\em chosen arbitrarily}}
\EndWhile

\EndIf

\EndIf 
\EndIf 

\EndWhile

\end{algorithmic}
\end{algorithm}

\begin{proof}
    \pjedit{It suffices to consider $k=k^*$, since $\gamma_k$ is increasing.}
    By Theorem~\ref{thm11} we have that \nonadapt\ is guaranteed to achieve approximation $\gamma_{\hat{k}}$ whenever $u\leq \frac{2}{k-1+\sqrt{(k-1)^2+4}}$. That is since for any $k\geq 3$, $\frac{2}{k-1+\sqrt{(k-1)^2+4}}\leq \frac{1}{k-1}\leq \frac{1-\gamma_{k-1}}{\gamma_{k-1}^2}
    \leq \frac{1-\gamma_{\hat{k}}}{\gamma_{\hat{k}}^2}$ (from properties in \cite{GourvesP24}\pareplace{, \eqref{gammakeq3} in Appendix~\ref{ap:DetailsFromGP}}{} and $\gamma_{k-2}<\gamma_{k-1}$). Finally for $k<3$, $\hat{k} = 1$ and \nonadapt\ gives \cra\ $\gamma_1$ for any $u$.

\pjmargincomment{Correction: deletion}\pjdelete{    Thus, it suffices to show that  {\adapt} will behave exactly like \nonadapt\ when run with parameter $\hat{k}$.}

As \adapt\ runs, it computes the optimal parameter $k_{curr}$ for \nonadapt\ based on the largest item seen so far, using Observation~\ref{obs:optimal_k}, and sets $k'=\max\{k_{curr}-2,1\}$ (i.e., \pjreplace{after an item of the maximum size $u$ is revealed}{when the largest item size revealed is $u$}, we have $k'=\hat{k}$). The algorithm then updates the set $B$ of non-small items according to \pjreplace{$\hat{k}$}{$k'$} and follows the same thresholds and removal rules as \nonadapt\ with parameter $k'$.

\pjmargincomment{Big proof simplification (next two paragraphs also)!!}\pjreplace{For the proof, let $k'_i$, for $i\in \mathbb{N}^*$, 
denote the value of $k'$ after the $i$th iteration (i.e., after receiving item $o_i$ and updating $k'$). We further denote by $B_i$ and $B_{k,i}$ the sets $B$ for \adapt\ and \nonadapt\ respectively, after the $i$th iteration, where the $k$ of $B_{k,i}$ is the parameter used in this instance of \nonadapt .

    Using the following claim, we show that for any $k'_i$, the set $B$ of non-small items in \adapt\ coincides with the corresponding set in \nonadapt\ run with parameter $k'_i$ on the same input. That is, $B_i = B_{k'_i,i}$:
    
}{We use the following claim. Who's proof is provided right after.}

    \begin{ourClaim}
\label{cl:noBigItemsDeleted}
    If the largest item size seen so far satisfies  $MaxItem\leq \frac{2}{4-1+\sqrt{(4-1)^2+4}}$ \armargincomment{why not $\frac{2}{3+\sqrt{13}}$ ?} (i.e., $k\geq 4$), then \adapt\ has never removed any item in set $B$ \pjedit{i.e., a non-small item.\footnote{That is, no item is deleted while being a member of $B$ i.e., considered non-small at the current phase. An item that has been a member of $B$ but is in a later phase re-classified as small (i.e., member of $S$ but not member  of $B$) can get deleted.}}
\end{ourClaim}
    
\pjreplace{    We can now complete the proof of the theorem by induction on $k'_i$: 
    
        \noindent\textbf{Base Case, $i=1$:} Since $k'_1$ is the first value for $k'$ selected by \adapt, the first iteration is exactly the same as for \nonadapt\ when running for $k'_1$.
        \textbf{Induction Step, $i>1$:} Suppose that the hypothesis is true for all $k'_j$, $1\leq j<i$. \newline
        \textbf{Case I, $k'_i=k'_{i-1}$:} That is the simplest case. Since we have $B_{i-1} = B_{k'_{i-1}, i-1}$ and $k'$ doesn't change at step $i$,  then \adapt\ and \nonadapt\ behave exactly identical with regards to non-small items. That is, $B_{k'_i,i} = B_{k'_{i-1},i} = B_i$. \\
        \textbf{Case II, $k'_i\neq k'_{i-1}$:} We identify the following sub-cases:

                \textbf{Case II.a, $k'_i \neq 1$:} We have $k_{curr}\geq 4$ since it must be $k'\geq 2$. By Claim~\ref{cl:noBigItemsDeleted}, after iteration $i$, \adapt\ has either terminated with approximation at least $\gamma_{k'}$ or has not removed any item from $B_i$. Items in $B_{i-1}$ are either reclassified as small at the start of iteration $i$ (and possibly deleted after) or remain in $B_i$. Thus, every item seen so far that is not $k'_i$-small remains in $B_i$.

We now consider \nonadapt\ and the set $B_{k',i}$. Since all non-$k'_i$-small items fit in the knapsack, the condition to remove such items is not met, and $B_{k',i} = B_i$.

                \textbf{Case II.b, $k'_i = 1$:} The case $k'_{i-1} = 1$ is covered by \textbf{Case I}, so we consider only $k'_{i-1} \neq 1$. By definition, this is the first iteration in which an item large enough to set $k' = 1$ appears. By induction, we have $B_{i-1} = B_{k'_{i-1}, i-1}$. After item $o_i$ arrives, \adapt\ updates $k'$ to $k'_i = 1$ (equal to $\hat{k}$, as it will not change again) and resets $B_i$ in Step~\ref{algStep:smallRemoval} to include all non-$1$-small items seen so far, none of which have been deleted by Claim~\ref{cl:noBigItemsDeleted} and induction. It then handles item $o_i$.

Finally, it suffices to note that $B_{1,i-1}$ contains exactly the same items as the set $B$ of \adapt\ immediately after the reset in Step~\ref{algStep:smallRemoval}. Since both algorithms then process the same item in iteration $i$, we obtain $B_{1,i} = B_i$.

\smallskip
\noindent This concludes the proof since the order of deletion of small items is arbitrary and therefore not important for the results guaranteed by \nonadapt. \qed
}{While $k'\ge2$, Claim~\ref{cl:noBigItemsDeleted} shows that no non-small item is removed before the desired \cra\ is reached. If the algorithm removes small items to restore feasibility, then just before the last removal the value is greater than $1$, and the removed item has size at most $1-\gamma_{k'}$. The remaining value is therefore greater than $\gamma_{k'}$. Since $k'$ can only decrease, this solution satisfies every subsequent threshold and is not changed afterwards. The same holds if line~\ref{adline291} produces a solution of value at least $\gamma_{k'}$.

Consequently, either the desired \cra\ has already been reached, or no item has been removed while $k'\ge2$. If the input ends in the latter case, the solution contains all the items and is optimal. If $k'$ becomes $1$, the solution still contains all the previously seen items, unless the desired \cra\ has already been reached. If its value is at least $\gamma_1$, the algorithm does not change it. Otherwise, \nonadapt\ run with parameter $1$ from the beginning would have kept the same items, and Algorithm~\ref{algo_adaptive} subsequently follows its rules, with the possible earlier selection of a subset of $B$ of value at least $\gamma_1$. Theorem~\ref{thm11} therefore gives \cra\ $\gamma_1$ in this case. In every case, the \cra\ is at least $\gamma_{\hat k}$. \qed
}
\end{proof}

\pjmargincomment{Also brought this here without lemmas 5 and 6, but I think they are pretty clear when mentione din the proof.}
\pjedit{
\begin{proof}[Of Claim~\ref{cl:noBigItemsDeleted}]
    We prove the claim by showing that because we are checking for a set $S$ satisfying an approximation ratio of $\gamma_{k'}$, instead of $\gamma_{k_{curr}}$\footnote{Where $k_{curr}$ the optimal $k$ for the items bounded by the highest value seen so far.}, the sum of items in $B$ will always become larger than $\gamma_{k'}$ before becoming infeasible (i.e., larger than 1). We show this by using Lemmas \ref{lm:sufficientNonSmallItems} and \ref{lm:largeItemsBound}, that are provided and proven in Appendix~\ref{ap:omittedAdaptiveRemoval}.
    
        Specifically, by definition, at any point, the algorithm has seen items of size up to \[
    \frac{2}{k_{curr}-1+\sqrt{(k_{curr}-1)^2+4}}.
    \]
    
    By step~\ref{algStep:smallRemoval} of the algorithm\pjreplace{ - removing all non-small items from $B$ whenever $k'$ is updated}{, resetting $B$ at an overflow}, all items in $B$ \pjreplace{at any point}{at that point} must be of size \pjreplace{at least}{strictly greater than} $1-\gamma_{k'} = 1-\gamma_{k_{curr}-2}$.
    
    By Lemma~\ref{lm:sufficientNonSmallItems} we have that 
    
    \[
    (k-1)(1-\gamma_{k-2})\geq \gamma_{k-2}\Rightarrow (k_{curr}-1)(1-\gamma_{k_{curr}-2})\pjreplace{>}{\ge}\gamma_{k_{curr}-2}=\gamma_{k'}
    \].

This means that whenever there are $k_{curr}-1$ items in $B$, there is a set $S\subseteq B$ that satisfies \pjreplace{$|S|\geq \gamma_{k'}$}{$v(S)>\gamma_{k'}$}.

Finally, by Lemma \ref{lm:largeItemsBound}, we have: 
\[
(k_{curr}-1)\frac{2}{k_{curr}-1+\sqrt{(k_{curr}-1)^2+4}}\leq 1.
\]

The right part of the product on the left side of the inequality is an upper bound on the largest possible item size seen so far. Thus, whenever there are $k_{curr}-1$ items in $B$, they can fit in the knapsack and would be returned as the algorithm's solution. \qed
\end{proof}
}

\pjreplace{\subsubsection{Learning-Augmented Variant}}{\subsection{Learning-Augmented Variant}}

 The adaptive algorithm also allows to create a \emph{learning-augmented via predictions version} (cf.  \cite{site:ALPS,predictions}).

The goal is to design algorithms that leverage predictive input, motivated by modern ML systems’ ability to provide such insight. We are interested in two performance measures: (a) \textbf{Consistency}, the worst-case \cra\ when predictions are perfect, and (b) \textbf{Robustness}, the worst-case \cra\ when predictions are arbitrarily incorrect. In our setting, the prediction concerns the true item-size upper bound.

For our version, we identify the following key property, whose proof is provided in Appendix~\ref{ap:predictions}:

\begin{lemma}
\label{lemma:k-sizes}
Items of size up to $\frac{2}{k-1+\sqrt{(k-1)^2+4}}$ are $\{k-c\}-small$, for $k > c\geq 3$. 
\end{lemma}

This implies that the deletion order of previously seen items is irrelevant when targeting a \pjreplace{$\gamma_{k'-3}$}{$\gamma_{k-3}$} approximation. This allows us to follow the non-adaptive strategy of \cite{GourvesP24} and, if the prediction is invalidated, switch to \adapt’s strategy. 

Specifically, the algorithm needs to follow the following basic steps:
\begin{enumerate}
    \item For the prediction of the item-size upper bound  $u_p\pjedit{\in(0,1]}$, let $k_p$ be the optimal value of $k$,  according to Observation~\ref{obs:optimal_k}.
    \armargincomment{slightly rephrased 1.}
    \item If $k_p<4$, then run \adapt\ until completion, else continue.
    \item Start by running algorithm \adapt .\label{step:adapt} 
    \item If at any step it holds that \pjreplace{$k' \in \{k_p, k_{p-1}, k_{p-2}\}$}{$k' \in \{k_p, k_p-1, k_p-2\}$}, then keep the current knapsack state as is and start running \nonadapt\ \pjedit{with parameter $k_p$} from there on.\label{step:nonAdapt}
    \item If \pjreplace{during step \ref{step:nonAdapt}}{at any step} an item $o$ arrives such that \pjreplace{$o>u_p$}{$|o|>u_p$}, then recalculate the new $k_{curr}$ (as in \pjreplace{step 6 of \adapt}{line~\ref{lineKcurr} of \adapt}), keep the current knapsack state as is, and start running \adapt\ setting  \pjreplace{$k' = \min\{k_{p-3}, k_{curr-2}\}$}{\[k'=\max\{1,\min\{k_p-3,k_{curr}-2\}\}.\]}\label{step:reAdapt}
\pjedit{Use this assignment in every subsequent iteration and do not use the prediction again.}
\end{enumerate}
\pjedit{The tests in steps~\ref{step:nonAdapt} and~\ref{step:reAdapt} are performed after updating $MaxItem$ and $k_{curr}$, but before processing the arriving item. Step~\ref{step:reAdapt} takes precedence over step~\ref{step:nonAdapt}.}

    \begin{theorem}
\label{thm:augmented}
There exists an algorithm that, given a prediction of the item-size upper bound, achieves \emph{Consistency} \pjreplace{equal to the \cra}{at least the guarantee}\ of \nonadapt\ from \cite{GourvesP24} \pjreplace{(i.e., approximately $\gamma_{k^*}$)}{with parameter $k^*$} and \emph{Robustness} of $\gamma_{k^*-3}$ for all $k^* \geq 4$. For $1 \leq k^* < 4$, it guarantees both \emph{Consistency} and \emph{Robustness} of $\gamma_{\hat{k}}$, where $\hat{k}=\max\{k^*-2,1\}$.
\end{theorem}

\pjedit{
\begin{proof}
\pjreplace{For the second case of the theorem, $k<4$, we notice that whenever}{Whenever} $k_p<4$, we only run \adapt\ and ignore the prediction. Thus, we guarantee approximation of $\gamma_{\hat{k}}$\pjedit{, where $\hat{k}=\max\{k^*-2,1\}$,} by Theorem~\ref{thm:adaptiveMain} which bounds both the robustness and the consistency of the algorithm.

We now focus on the more interesting case where $k_p\geq 4$. If $k_{opt}$ is the optimal $k$ for the real upper bound $u$ according to Observation~\ref{obs:optimal_k}, we break the analysis into the following cases:

\begin{description}
    \item[Case $k_{opt}> k_p$.] In this case the prediction is not accurate and thus we need to show the robustness bound. We identify two sub-cases:
    \begin{description}
        \item[\pjreplace{$k_{opt-2}> k_{p}$:}{$k_{opt}-2> k_{p}$:}] Then, the algorithm will terminate during the run of \adapt\ in step \ref{step:adapt} guaranteeing approximation \pjreplace{$\gamma_{\hat{k}}>\gamma_{k-3}$}{$\gamma_{\max\{k_{opt}-2,1\}}>\gamma_{k_{opt}-3}$} (Theorem~\ref{thm:adaptiveMain}).
        \item[\pjreplace{$k_{opt-2}\leq k_{p}$:}{$k_{opt}-2\leq k_{p}$:}] \pjreplace{Then the algorithm will terminate during the run of Algorithm~\ref{algo7} at step\ref{step:nonAdapt}.}{Then the algorithm reaches step~\ref{step:nonAdapt}, unless the required guarantee has already been achieved.} Notice that our algorithm will immediately switch to Algorithm~\ref{algo7} as soon as $k'\leq k_{p}$.\pjedit{ This switch is performed before processing the arriving item. Before this arrival, $k_{curr}\geq k_p+3$, unless this is the first item.} That is, at the point of the switch, all items \pjreplace{seen so far}{preceding the arriving item} would be considered small, w.r.t to $k_p$, by Lemma~\ref{lemma:k-sizes}. Thus, using the same logic as in the proof of Theorem~\ref{thm:adaptiveMain}, all \pjreplace{items seen so far}{earlier items} would have been considered as $k_p-small$ by Algorithm~\ref{algo7} and their deletion order doesn't matter. \pjreplace{That is, until this point, the algorithm has behaved exactly as if we were running Algorithm~\ref{algo7} from the beginning.}{More precisely, by the proof of Theorem~\ref{thm:adaptiveMain}, any earlier deletion or rejection leaves value at least $\gamma_{k_p}$, since until the switch $k'>k_p$. Otherwise, all earlier items remain. Thus, unless the value already reaches $\gamma_{k_p}$, the current state is the same as for Algorithm~\ref{algo7} with parameter $k_p$ from the beginning. The arriving item is then processed by Algorithm~\ref{algo7}.} Since $k_p<k_{opt}$ \pjreplace{(i.e., $u_p>u$)}{(hence $u_p>u$)}, Algorithm~\ref{algo7} will still provide approximation $\gamma_{k_p}$, by Theorem~\ref{thm11} \cite{GourvesP24}. But we have \pjreplace{$k_p\geq k_{opt-2}$}{$k_p\geq k_{opt}-2$} which is better than the desired robustness.
    \end{description}
    \item[Case $k_{opt}=k_p$.] \pjreplace{I.e., the prediction agrees with the real upper-bound and we need to show the consistency bound.}{First suppose that the prediction is never invalidated.} This case follows \pjreplace{exactly the same logic}{the same switching argument} as the sub case \pjreplace{($k_{opt-2}\leq k_{p}$)}{($k_{opt}-2\leq k_{p}$)} of the previous case. All we have to notice is that since now $k_{opt}=k_p$, Algorithm~\ref{algo7} will guarantee its best possible competitive ratio as stated by Theorem~\ref{thm11} \pareplace{and Theorem~\ref{thm111}}  (mentioned in Section~\ref{sec:adaptive}) of \cite{GourvesP24}.\pjedit{ This proves consistency when $u_p=u$, and robustness otherwise, since this guarantee is at least $\gamma_{k_p-1}\geq\gamma_{k_p-3}$. If the prediction is invalidated, the argument in the following case applies as well.}
\item[\pjreplace{Case $k_{opt}>k_p$.}{Case $k_{opt}<k_p$.}] I.e., the prediction is not accurate and thus we need to show the robustness bound. \pjreplace{The only case where the procedure of \adapt\ would be disrupted, is if at any step, $k' \in \{k_p, k_{p-1}, k_{p-2}\}$, thus entering step \ref{step:nonAdapt} of the algorithm.}{An item invalidating the prediction must arrive. The following argument also applies when $k_{opt}=k_p$ and the prediction is invalidated.} \pjreplace{But, if at any point after that, an item}{When the first item} invalidating the prediction arrives (i.e., an item of size $|o|>u_p$)\pjreplace{, and the current termination bound isn't reached, then we immediately return to making an over-estimation of $k$, by setting $k'< k_{p-3}$.}{, we immediately return to running \adapt\ with $k'=\max\{1,\min\{k_p-3,k_{curr}-2\}\}$ at every subsequent iteration.} By Lemma~\ref{lemma:k-sizes}, this means that \pjreplace{all previously seen items are now $k_{p-3}$-small}{all items preceding this arrival are $\{k_p-3\}$-small}, so their deletion order doesn't matter when we are targeting an approximation ratio of at most \pjreplace{$\gamma_{k_{p-3}}$}{$\gamma_{k_p-3}$}.\pjedit{ More precisely, if an earlier deletion has occurred, the last deletion restoring feasibility leaves value greater than $\gamma_{k_p-3}$. Selection of a feasible subset or rejection at a termination condition also leaves at least this value. This already suffices, since $k_{opt}\leq k_p$. Otherwise, all earlier items remain. The proof of Theorem~\ref{thm:adaptiveMain} applies to the subsequent run because $k'$ never increases and, whenever $k'>1$, satisfies $k'\leq k_{curr}-2$: by Lemmas~\ref{lm:sufficientNonSmallItems} and~\ref{lm:largeItemsBound}, $k'+1$ non-small items fit and already reach $\gamma_{k'}$, so no non-small item is removed before reaching the target; deleting small items to restore feasibility reaches the target. The case $k'=1$ is handled as in that proof. The final guarantee is therefore $\gamma_{\max\{1,\min\{k_p-3,k_{opt}-2\}\}}$, which is at least $\gamma_{k_{opt}-3}$ when $k_{opt}\geq4$, and at least $\gamma_1$ otherwise.} Thus the desired guarantee directly follows from the analysis of \adapt\ (i.e., proof of Theorem~\ref{thm:adaptiveMain}).
\end{description}
This concludes the proof. \qed

\end{proof}

}

\subsection{Adaptive Algorithm with Removal and Recourse}
\label{sec:adaptRecourse}

In this section, we extend the previous result by turning our attention to adaptive algorithms (with unknown item size upper bound) that can also use a recourse.

\begin{algorithm}[t]
\caption{\adaptRec 
\label{algo_adaptive_recourse}}
\begin{algorithmic}[1] 
\State $S \gets \emptyset$
\State $MaxItem \gets 0$
\State $B \gets \emptyset$
\While{a new item $o_i$ arrives}
\State $MaxItem \gets \max\{MaxItem, |o_i|\}$

    
\State $k_{curr}\gets \max\Big\{k\in \mathbb{N}\text{ } \big | \text{ } \frac{2}{k-1+\sqrt{(k-1)^2+4}}\geq MaxItem \Big\}$ \Comment{Optimal $k$ considering $MaxItem$ as the upper bound.}

\State $k'\gets \max\{k_{curr}-2, 1\}$
\State $k''\gets \max\{k_{curr}-1, 1\}$



    \If{ \label{adRecline1024} $v(S) \ge \gamma_{k''}$}
        \State \label{adRecline1025}  Reject $o_i$ \Comment{{\em $S$ is not changed afterwards}}

\Else
\State \label{adRecline1027} $S \gets S \cup \{o_i\}$
    
\If{\label{adRecline1030} $v(S) > 1$}


\State Let $B=\{o \in S : |o| > 1-\gamma_{k'} \}$
\label{algStep:smallRemovalRec} \Comment{Reset B}

    \If{$S$ contains a subset $\hat S$ 
such that $1 \ge v(\hat S) \ge \gamma_{k''}$}


\State \label{adRecline291} $S \gets \hat S$ \Comment{{\em $S$ is not changed afterwards}}
\BREAK
\EndIf

\If{  $v(B) > 1$ \pjcomment{Had missed that!}\pjedit{\textbf{or} $B=\varnothing$}}

\State \label{adRecline1044} Remove the largest item of $S$ (send to the buffer) 


\ElsIf{$|B|=1$ ${\wedge}$ $\exists\ o'\in \textit{Deleted Items}$ s.t.: \pjreplace{$S\bigcup o' \setminus B $}{$(S\cup\{o'\})\setminus B$} contains a 
\Statex \hspace{18mm}subset $\hat S$ such that $1 \ge v(\hat S) \ge \gamma_{k''}$}

\State Use the recourse to bring item $o'$ back and set:

\State \label{adRecline291b} $S \gets \hat S$ \Comment{{\em $S$ is not changed afterwards}}
\BREAK

\Else
\While{$v(S)>1$ \label{adRecline1050}}
\State \label{adRecline1051} Remove from $S$ the smallest element of $S \setminus B$ (send to the \Statex \hspace{25mm} buffer)
\EndWhile

\EndIf 

\EndIf 
\EndIf 

\EndWhile

\end{algorithmic}
\end{algorithm}

 The goal is to incorporate the recourse into a new version of the adaptive algorithm so that we can create an equivalent of Theorem~\ref{thm:adaptiveMain}, that guarantees a {\cra} even closer to the best possible that the non adaptive algorithm can achieve. Specifically, we aim for a {\cra} 
 of at least $\gamma_{\hat{k}+1}$ \pjreplace{(in the general case)}{when $k^*\geq 3$}. The main logic is as follows:

Maintain the same strategy as \adapt\ and keep considering items of size \pjedit{larger than} $1-\gamma_{k'}$ to be non-small, where $k' = \max\{\pjreplace{k_{curr-2}}{k_{curr}-2}, 1\}$, but only terminate if the sum of collected items gives \cra\ (at least) \pjreplace{$\gamma_{k'+1}$}{$\gamma_{k''}$, where $k''=\max\{k_{curr}-1,1\}$}. That is, run the algorithm as if we were trying to get \cra\ $\gamma_{\hat{k}}$, but only accept the solution if a \cra\ of $\gamma_{\hat{k}+1}$ is satisfied instead. While doing so, also keep track of the items of size larger than \pjreplace{$1-\gamma_{k'+1}$}{$1-\gamma_{k''}$}, i.e., \emph{``non-very-small items''}. I.e., the items we should have considered as ``non-small'' if we were aiming for \cra\ of $\gamma_{\hat{k}+1}$, by the logic of \nonadapt . 



Unlike non-small items in the proof of Theorem~\ref{thm:adaptiveMain}, which always additively guarantee the desired {\cra} 
before becoming infeasible, non-very-small items lack this property. As a result, the adaptive algorithm may remove items that it could later combine with new arrivals to achieve the target ratio. With careful analysis, we identify the unique case where this occurs and show that the algorithm can \pjreplace{immediately detect and reverse it with a single recourse}{detect it and use a single recourse when needed}.


\pjmargincomment{Updated for missed case}\pjreplace{Refining the algorithm to delete small items in increasing order—without loss of generality, since arbitrary deletion sufficed before—we observe the following:



\emph{If any subset of the items seen so far is feasible and can yield an approximation of at least $\gamma_{k'+1}$ (for any current value of $k'$), then either there is a feasible subset among the items currently in the solution that satisfies this, or it can be achieved by a single swap between a rejected item and an item in the solution.}
}{When $B$ is empty, the algorithm removes the largest item, as
\nonadapt\ would do. Otherwise, it deletes small items in increasing
order. If the two algorithms delete different items before reaching the target,
Algorithm~\ref{algo_adaptive_recourse} retains at least as much value as \nonadapt. Upon the next
non-very-small item, either a target solution is available among the current
items, or it can be obtained by a single recourse.}
\pjdelete{\pjreplace{\pareplace{We state the main theorem here and provide the related algorithm, {\adaptRec}  (Algo.~\ref{algo_adaptive_recourse}), as well as the technical analysis in Appendix~\ref{ap:adaptRec}}{The technical analysis is deferred to the full version of the paper.}}{The technical analysis is given in Appendix~\ref{ap:adaptRec}.}}

\begin{theorem}
    \label{Thm:Adaptive_recourse}
        For all positive integers $k$, Algorithm~\ref{algo_adaptive_recourse} is $\gamma_{\tilde{k}}$-competitive when $u\leq \frac{2}{k-1+\sqrt{(k-1)^2+4}}$, where \pjmargincomment{Same as previous theorem statement (!or not?!!)}$\tilde{k}= \max\{\pjreplace{k^*}{k}-1, 1\}$, \pareplace{}{using at most one recourse}.  
\end{theorem}

\pjedit{
 For the analysis, apart from the previous definition of the set of non-small items, $B$ (and its instance $B_i$ at step $i$), we also use $C$ for the set of non-very-small items and $C_i$ for its instance at step $i$. By definition, $B_i \subseteq C_i$ at any $i$. Even though Algorithm~\ref{algo_adaptive_recourse} still defines the set of non-small items to generally be $k-2$-non-small items, in this case we want to compare its behavior to the non-adaptive algorithm, Algorithm~\ref{algo7}, run with parameter $k-1$ (since we are trying to guarantee $\gamma_{k-1}$ approximation in this case).
\pjedit{Here $k=k_{curr}$ within the phase under consideration, and $B_i$ and
$C_i$ refer to the corresponding subsets of the currently held items.}

\begin{observation}
\label{obs:noDeletionofNonSmall}
    By Lemma~\ref{lm:sufficientNonSmallItems}, we have that the main condition used in the analysis of Algorithm~\ref{algo_adaptive}, i.e., that no item in the current set of non-small items $B_i$ will be deleted, also stands for Algorithm~\ref{algo_adaptive_recourse} \pjedit{before reaching the target, when $k_{curr}\geq3$}, where we are trying to reach $\gamma_{k-1}$ approximation.
\end{observation}

The proof is identical to Claim~\ref{cl:noBigItemsDeleted}. All we have to notice is that in that specific proof we didn't use  Lemma~\ref{lm:sufficientNonSmallItems} tightly: We only needed that  $(k-1)(1-\gamma_{k-2})\geq \gamma_{k-2}$ when it is also $(k-1)(1-\gamma_{k-2})\geq \gamma_{k-1}$.
\pjedit{The additional largest-item deletion in line~18 is performed only when
$B=\varnothing$, so it does not remove a non-small item.}

The following lemma is key to understanding the structure of non-small items. To maintain the narrative, we state the Lemma here and provide its proof in Appendix~\ref{ap:adaptRec}:

\begin{lemma}
    \label{lm:itemStructure}
    For all integers $k\geq 3$ and $0\leq n < k-2$,
    \[
    1>(k-n-1)(1-\gamma_{k-2})+n(1-\gamma_{k-1})\geq \gamma_{k-1}.
    \]
\end{lemma}

\begin{corollary}
\label{col:itemsCombination}
    Lemma~\ref{lm:itemStructure} in combination with Lemma~\ref{lm:largeItemsBound} show that whenever there are $k-1$ non-very-small-items such that at least two of them are non-small, then their sum is a feasible solution of at least $\gamma_{k-1}$ approximation ratio.
\end{corollary}

{\begin{observation}\label{obs:nonEmpty}
If Algorithm~\ref{algo_adaptive_recourse} deletes an item while $k_{\mathrm{curr}}=k\ge3$ and has not reached
its target, its solution after that iteration has value greater than
$\gamma_{k-2}$. If $k_{\mathrm{curr}}$ subsequently decreases, the new target is
already satisfied.
\end{observation}
\begin{proof}
By Observation~\ref{obs:noDeletionofNonSmall}, all deleted items have size at most
$1-\gamma_{k-2}$. Immediately before the last deletion that restores feasibility,
the solution has value greater than $1$. Therefore, afterwards its value is
greater than $\gamma_{k-2}$. Further iterations with the same parameter preserve
this bound. If the parameter decreases to $j\le k-1$, the new target is
$\gamma_{\max\{j-1,1\}}\le\gamma_{k-2}$, so line~9 applies before processing the
new item.
\end{proof}}

The following lemma completes the picture of the possible structures of non-small-items that could provide the approximation ratio. Its proof is provided in Appendix~\ref{ap:adaptRec}:

\begin{lemma}
    \label{lm:nonVerSmallBound}
    For any $k\geq 2$,
    \[
    1>k(1-\gamma_{k-1})\geq \gamma_{k-1}
    \]
\end{lemma}

That is, if $k$ non-very-small items fit in the knapsack, then their sum gives the required approximation. 

{The proof of the main theorem distinguishes the case where $B$ is empty
from the case where there is exactly one non-small item. In the former case,
Algorithm~\ref{algo_adaptive_recourse} removes the largest item, as Algorithm~\ref{algo7} would do. In the latter
case, we compare the values of the two solutions and show that the next
non-very-small item either gives a target solution directly or allows the use
of one recourse. By Observation~\ref{obs:nonEmpty}, \pjreplace{a change in $k_{\mathrm{curr}}$ before termination can occur only before any deletion} {any decrease in $k_{\mathrm{curr}}$ after a deletion immediately satisfies the new target}.}

We are now ready to conclude with the formal proof of Theorem~\ref{Thm:Adaptive_recourse}:

\begin{proof}[Of Theorem~\ref{Thm:Adaptive_recourse}]
Since $\gamma_k$ increases with $k$, it suffices to prove the guarantee for
$k=k^*$ of Observation~\ref{obs:optimal_k}. We first consider a phase during which
$k_{\mathrm{curr}}=k\ge3$. Write
\[
\beta=\gamma_{k-1},\qquad a=1-\gamma_{k-2},\qquad c=1-\beta.
\]
Thus $a>c$, $B=\{o\in S:|o|>a\}$ and $C=\{o\in S:|o|>c\}$.
Algorithm~\ref{algo7} run with parameter $k-1$ guarantees approximation $\beta$ for the
item-size upper bound of this phase, by Theorem~\ref{thm11} and the
same argument used in the proof of Theorem~\ref{thm:adaptiveMain}.

Suppose that an item causes $v(S)>1$, but the test in line~15 finds no subset
with value in $[\beta,1]$. Then $v(C)>1$. Otherwise, deleting items outside $C$,
each of size at most $c=1-\beta$, until feasibility is restored would leave a
subset of value greater than $\beta$, a contradiction. Just before the new item
arrived, the solution contained at most $k-1$ items in $C$: a feasible set of
$k$ such items would already reach the target by Lemma~\ref{lm:nonVerSmallBound}.
By Lemma~\ref{lm:largeItemsBound}, any $k-1$ items fit. Hence the current set $C$ has
exactly $k$ items. By Corollary~\ref{col:itemsCombination}, at most one of them belongs to
$B$, since otherwise a subset of $k-1$ items would reach the target.

If $B=\varnothing$, Algorithm~\ref{algo_adaptive_recourse} removes the largest item. This is exactly the
choice of Algorithm~\ref{algo7} when $v(C)>1$. It restores feasibility, since the removed
item is at least as large as the arriving item. Until the first deletion with
$B\ne\varnothing$, the two algorithms therefore have the same solution, unless
Algorithm~\ref{algo_adaptive_recourse} has already reached its target.

Consider that first deletion with $B\ne\varnothing$, and let $z>a$ be the size
of its unique non-small item. Let $d_0$ be the size of the smallest other item
in $C$, let $R$ be the total size of the remaining $k-2$ items in $C$, and let
$w$ be the total size of the items outside $C$. Algorithm~\ref{algo_adaptive_recourse} deletes all items
outside $C$ and then the item of size $d_0$, leaving value $z+R$; Algorithm~\ref{algo7}
removes the item of size $z$, leaving value $d_0+R+w$.
The former solution is feasible by Lemma~\ref{lm:largeItemsBound}. The latter is
feasible because $z$ is the largest item and the solution was feasible before
the arrival. Neither has value at least $\beta$, since both are subsets tested
in line~15.

In fact, the failed test also implies $z+R+w<\beta$. If this were false,
starting with the feasible set of value $z+R<\beta$ and adding items outside
$C$ until reaching $\beta$ would produce a subset of value at most
$\beta+c=1$. Since $d_0$ is the smallest of the $k-1$ items in $C\setminus B$,
we have $R-d_0\ge(k-3)c$. By Lemma~\ref{lm:itemStructure},
\[
2a+(k-3)c\ge\beta.
\]
Consequently,
\[
w<\beta-z-R\le z-d_0.
\]
Thus $z+R>d_0+R+w$: Algorithm~\ref{algo_adaptive_recourse} retains at least as much value as Algorithm~\ref{algo7}.
If the input ends now, its competitive ratio is at least $\beta$. Until another
item in $C$ arrives, items of size at most $c$ fit in both solutions while their
values are below $\beta$, and preserve this comparison.

Consider the next item of size greater than $c$. If a target subset exists,
Algorithm~\ref{algo_adaptive_recourse} finds it. Otherwise the preceding argument again gives exactly
$k$ items in $C$, with the same unique non-small item of size $z$. A second
non-small item would yield a target subset by Corollary~\ref{col:itemsCombination}.
Let $d_1$ be the smallest size among the other $k-1$ items in $C$, and let $R'$
be the sum of the remaining $k-2$ sizes. The set of value $z+R'$ is feasible
and has value less than $\beta$, so
\[
R'<\beta-z,\qquad d_1\le\frac{R'}{k-2}<\frac{\beta-z}{k-2}.
\]
The previous deletion similarly gives
$d_0\le R/(k-2)<(\beta-z)/(k-2)$. By restoring the item of size $d_0$ and
removing the non-small item, we obtain $k$ items in $C$ with total size
\begin{equation}\label{eq:recourse-corrected-sum}
R'+d_1+d_0<\frac{k(\beta-z)}{k-2}
<\frac{k(\beta-a)}{k-2}\le1.
\end{equation}
To verify the last inequality, Lemma~\ref{lm:itemStructure} gives
$2a+(k-2)c\ge1$. For $k\ge4$, using $a>c$, we obtain
\[
k(a+c)=2\bigl(2a+(k-2)c\bigr)+(k-4)(a-c)\ge2.
\]
Since $\beta-a=1-(a+c)$, the inequality follows. For $k=3$, it follows directly
from
\[
a+c=\frac{5-\sqrt5-\sqrt2}{2}>\frac23.
\]
By Lemma~\ref{lm:nonVerSmallBound}, the feasible set in~\eqref{eq:recourse-corrected-sum}
has value greater than $kc\ge\beta$. It is available in the test of line~20
by using the previously deleted item of size $d_0$. The algorithm therefore
uses one recourse and reaches its target. No further recourse is used.

It remains to consider changes in $k_{\mathrm{curr}}$. By
Observation~\ref{obs:nonEmpty}, a decrease after any deletion immediately
satisfies the new target. Every change before that has retained all items seen
so far, so the preceding comparison with Algorithm~\ref{algo7} begins from the same
solution. If $k_{\mathrm{curr}}\in\{1,2\}$, then $k'=k''=1$. After a failed target
subset test, $v(B)>1$ necessarily holds, by the same small-item deletion
argument. The algorithm then behaves as Algorithm~\ref{algo7} with parameter $1$ and
guarantees $\gamma_1$.

Finally, a solution that reaches its current target also meets every later
target, since $k_{\mathrm{curr}}$ can only decrease. If the input ends before
this happens, the comparison with Algorithm~\ref{algo7} in the final phase gives the
claimed ratio. This completes the proof.
\end{proof}

}

In Figure~\ref{figAlgosComp}, we show the performance of, {\adapt} and {\adaptRec}, compared to the upper-bound aware algorithm of \cite{GourvesP24}.

\begin{figure}[H]
\center
\begin{tikzpicture}[scale=6]

\draw[->] (-0.03,0.5)--(1.05,0.5);
\draw[->] (0,0.47)--(0,1.05);


\draw  (1,0.47)--(1,0.53);
\draw  (0.03,1)--(-0.03,1);
\draw (-0.03,0.618033989)--(0.03,0.618033989);
\draw (0.618033989,0.53)--(0.618033989,0.47);
\draw (0.5,0.53)--(0.5,0.47);



\node () at (1.05,0.45) {$u$};
\node [rotate=90] () at (-0.05,1.05) {{\cra}};
\node () at (-0.1,1) {$1$};
\node () at (0.618033989,0.4) {$\phi^{-1}$};
\node () at (1,0.4) {$1$};
\node () at (-0.18,0.618033989) {$\phi^{-1}\approx 0.618$};
\node () at (0.5,0.4) {$0.5$};
\node () at (-0.1,0.5) {$0.5$};

\draw plot [domain=0:0.076118446,variable=\t] (\t,{1/(1+\t)});

\draw plot [domain=0.585786438:0.618033989,variable=\t] (\t,{(1-\t)/\t});
\draw plot [domain=0.394448725:0.414213562,variable=\t] (\t,{(1-\t)/(2*\t)});
\draw plot [domain=0.291796068:0.302775638,variable=\t] (\t,{(1-\t)/(3*\t)});
\draw plot [domain=0.229670386:0.236067977,variable=\t] (\t,{(1-\t)/(4*\t)});
\draw plot [domain=0.192582404:0.188611699,variable=\t] (\t,{(1-\t)/(5*\t)});
\draw plot [domain=0.16227766:0.159670332,variable=\t] (\t,{(1-\t)/(6*\t)});
\draw plot [domain=0.140054945:0.138260621,variable=\t] (\t,{(1-\t)/(7*\t)});
\draw plot [domain=0.123105626:0.121822171,variable=\t] (\t,{(1-\t)/(8*\t)});
\draw plot [domain=0.108824378:0.109772229,variable=\t] (\t,{(1-\t)/(9*\t)});
\draw plot [domain=0.099019514:0.098300563,variable=\t] (\t,{(1-\t)/(10*\t)});
\draw plot [domain=0.090169944:0.089612167,variable=\t] (\t,{(1-\t)/(11*\t)});
\draw plot [domain=0.08276253:0.082321372,variable=\t] (\t,{(1-\t)/(12*\t)});
\draw plot [domain=0.076473219:0.076118446,variable=\t] (\t,{(1-\t)/(13*\t)});


\draw (0.618033989,0.618033989)--(1,0.618033989);
\draw (0.414213562,0.707106781)--(0.585786438,0.707106781);
\draw (0.394448725,0.767591879)--(0.302775638,0.767591879);
\draw (0.291796068,0.809016994)--(0.236067977,0.809016994);
\draw (0.229670386,0.838516481)--(0.192582404,0.838516481);
\draw (0.188611699,0.86037961)--(0.16227766,0.86037961);
\draw (0.188611699,0.86037961)--(0.16227766,0.86037961);

\draw (0.159670332,0.877150706)--(0.140054945,0.877150706);
\draw (0.138260621,0.890388203)--(0.123105626,0.890388203);
\draw (0.121822171,0.901085803)--(0.109772229,0.901085803);
\draw (0.108824378,0.909901951)--(0.099019514,0.909901951);
\draw (0.098300563,0.917288177)--(0.090169944,0.917288177);
\draw (0.089612167,0.923563544)--(0.08276253,0.923563544);
\draw (0.082321372,0.928959478)--(0.076473219,0.928959478);
\draw (0.076118446,0.933647701)--(0.071067812,0.933647701);


\draw[purple,densely dashed, domain=0:0.068,variable=\t] plot (\t,{1 - 1.15*\t)});
\draw[purple,densely dashed, line width=1.02] (0.302775638,0.618033989)--(1,0.618033989);
\draw[purple,densely dashed,] (0.302775638,0.707106781)--(0.229670386,0.707106781);
\draw[purple,densely dashed,] (0.188611699,0.767591879)--(0.229670386,0.767591879);
\draw[purple,densely dashed,] (0.159670332,0.809016994)--(0.188611699,0.809016994);
\draw[purple,densely dashed,] (0.138260621,0.838516481)--(0.159670332,0.838516481);
\draw[purple,densely dashed,] (0.121822171,0.86037961)--(0.138260621,0.86037961);
\draw[purple,densely dashed,] (0.108824378,0.877150706)--(0.121822171,0.877150706);
\draw[purple,densely dashed,] (0.098300563,0.890388203)--(0.108824378,0.890388203);
\draw[purple,densely dashed,] (0.089612167,0.901085803)--(0.098300563,0.901085803);
\draw[purple,densely dashed,] (0.0823213727,0.909901951)--(0.089612167,0.909901951);
\draw[purple,densely dashed,] (0.076118446,0.917288177)--(0.0823213727,0.917288177);

\draw[blue, dotted,line width=1.05pt, domain=0:0.068,variable=\t] plot (\t,{1 - 1.05*\t)});
\draw[blue, dotted, line width=1.3pt, ] (0.414213562,0.618033989)--(1,0.618033989);
\draw[blue, dotted,line width=1.1pt] (0.302775638,0.707106781)--(0.414213562,0.707106781);
\draw[blue, dotted,line width=1.05pt] (0.229670386,0.767591879)--(0.302775638,0.767591879);
\draw[blue, dotted,line width=1.05pt] (0.188611699,0.809016994)--(0.229670386,0.809016994);
\draw[blue, dotted,line width=1.05pt] (0.159670332,0.838516481)--(0.188611699,0.838516481);
\draw[blue, dotted,line width=1.05pt] (0.138260621,0.86037961)--(0.159670332,0.86037961);
\draw[blue, dotted,line width=1.05pt] (0.121822171,0.877150706)--(0.138260621,0.877150706);
\draw[blue, dotted,line width=1.05pt] (0.10882437,0.890388203)--(0.121822171,0.890388203);
\draw[blue, dotted,line width=1.02pt] (0.098300563,0.901085803)--(0.10882437,0.901085803);
\draw[blue, dotted,line width=1.02pt] (0.089612167,0.909901951)--(0.098300563,0.909901951);
\draw[blue, dotted,line width=1.02pt] (0.0823213727,0.917288177)--(0.089612167,0.917288177);

\end{tikzpicture}
\caption{The black continuous lines represent the \cra\ of \nonadapt , the red dashed the \cra\ of \adapt\ and the blue dotted the \cra\ of \adaptRec.\label{figAlgosComp}}
\end{figure}
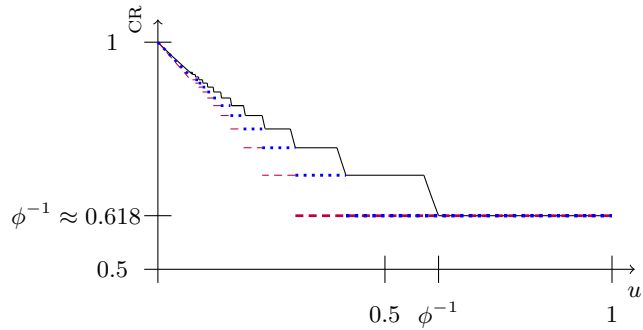

\begin{credits}
\subsubsection{\ackname} 
This work has been partially supported by project MIS 5154714 of the National Recovery and Resilience Plan Greece 2.0 funded by the European Union under the NextGenerationEU Program. 
Panagiotis Patsilinakos has received support
under the program ``Investissements d’Avenir'' launched
by the French Government, project PSL JRC 2024 - N 2024-
488. This work received the support of ``Partenariat Hubert Curien France/Gr{\`e}ce -- ARISTOTE''. 
\end{credits}

%
%
%
 \bibliographystyle{splncs04}

\bibliography{references}

\pagebreak

\section{Appendix}

\subsection{Proof of Observation \ref{intermobs}}\label{app:obs1}

\noindent \textbf{Observation \ref{intermobs}} \emph{$v(S_j^\ell) \ge \frac{k+1}{k+2} v(Q_j)$  implies 
$v(S_j) \ge \frac{k+1}{k+2} v(OPT_{j})$.}
\begin{proof}
Let $OPT_i^s$ and $OPT_i^\ell$ be the restrictions of $OPT_i$ to small and large items, respectively. 
We have 
\begin{equation} \label{decomposition_gen}
v(OPT_j)=v(OPT_j^s)+v(OPT_j^\ell) \mbox{ and } v(S_j)=v(S_j^s)+v(S_j^\ell)
\end{equation}
for $j=1..i$. By the definition of $Q_i$, we know that 
\begin{equation} \label{ideal_opt_gen}
v(Q_j) \geq v(OPT_j^\ell) 
\end{equation}
holds for $j=1..i$. As previously mentioned, the algorithm puts every small item in the current solution until we possibly have $v(S_i) \ge \frac{k+1}{k+2}$, meaning that
\begin{equation} \label{sbtoptfsi_gen}
v(S_{j}^s) \ge v(OPT_{j}^s)
\end{equation} 
holds for $j=1..i-1$. 
Combining $(\ref{recinvariant_gen})$,  $(\ref{decomposition_gen})$,  $(\ref{ideal_opt_gen})$, and $(\ref{sbtoptfsi_gen})$ gives
$v(S_{j}) =v(S_{j}^s)+v(S_{j}^\ell) \ge v(OPT_{j}^s) + \frac{k+1}{k+2} \pjreplace{v(Q_{j}^\ell)}{v(Q_{j})}  \ge   v(OPT_{j}^s)+\frac{k+1}{k+2} v(OPT_{j}^\ell) \ge \frac{k+1}{k+2} v(OPT_{j}).$ \qed
\end{proof}

\pjdelete{
\subsection{Correctness of Algorithm \ref{algorec4}}\label{app:correct1}

In this section we show that Algorithm \ref{algorec4} always produces a solution of value at most~$1$, and that at most $\lfloor k/2\rfloor$ items are taken from the buffer.    

\begin{observation} \label{obs_correctness} Algorithm \ref{algorec4} is correct.  
\end{observation}

\begin{proof}


Let us first consider how many times the recourse is used. By construction, recourse is made for possibly multiple items but during a single round, and for each recourse a solution whose value is between $\frac{k+1}{k+2}$ and $1$ is produced.  

At line \ref{linecalT_gen}, at most $\lfloor k/2\rfloor$ items are taken from the buffer. 
No recourse is necessary for building $S_i$ from $S_{i-1}$ and $o_i$ at line \ref{line2396_gen}. Indeed, we have seen in the proof of Theorem \ref{thmk1k2} that $S_{i-1}$ contains exactly $k$ large items. More precisely,  
$S_{i-1}$ contains the 
$\lceil \frac{k}{2}\rceil$ smallest large items, say $\{a_1, \ldots,a_{\lceil \frac{k}{2}\rceil}\}$, and the 
$\lfloor \frac{k}{2}\rfloor$ largest large items, say $\{b_1, \ldots,b_{\lfloor \frac{k}{2}\rfloor}\}$, of $\{o_1, \ldots,o_{i-1}\}$, 
where $|a_1| \le \ldots \le |a_{\lceil \frac{k}{2}\rceil}| \le |b_1| \le \ldots \le |b_{\lfloor \frac{k}{2}\rfloor}|$. If $|o_i| < |a_{\lceil \frac{k}{2}\rceil}|$, then $S_{i}=S_{i-1} \Delta \{o_i,a_{\lceil \frac{k}{2}\rceil}\}$. If $|o_i| > |b_{1}|$, then $S_{i}=S_{i-1} \Delta \{o_i,b_1\}$. If $|a_{\lceil \frac{k}{2}\rceil}| \le |o_i| \le |b_{1}|$, then $S_{i}=S_{i-1}$. 

It remains to prove that the algorithm maintains a solution whose value does not exceed 1. It is immediate at lines \ref{lineincl2_gen}, \ref{linetrim3_gen} and \ref{linetef2_gen}. Let us show that $v(S_i) \le 1$ holds at line \ref{line2396_gen}. We have $v(S_i^\ell) < \frac{k+1}{k+2}$ because $S_i^\ell \in {\cal T}$. We also know that $v(S_i^s)\le \frac{1}{k}-\frac{1}{k+2}$ because $v(S_{i}^s)=v(S_{i-1}^s) = v(S_{i-1})-v(S_{i-1}^\ell) < \frac{k+1}{k+2} - v(S_{i-1}^\ell)$, and  
$v(S_{i-1}^\ell)>1-|o_i|\ge 1-1/k$.  Thus, it holds that $v(S_i) = v(S_i^s) + v(S_i^\ell) \le \frac{k+1}{k+2}+\frac{1}{k}-\frac{1}{k+2} \le 1$ when $k \ge 2$. \qed
\end{proof}
}

\pjdelete{
\subsection{Complementary Upper Bounds}\label{app:upper}

\noindent {\textbf{Lemma \ref{lemmaub1rec3}.} \it
 For every integer $k \ge 2$, there exists an upper bound of $\frac{2k}{2k+1}$ on the {\cra}  of deterministic algorithms with at most $k-1$ recourse operations when $\frac{2}{2k+1} < u \le 1$.}
\begin{proof} 

Take any online algorithm and suppose that an adversary discloses the following instance:  $k$ \emph{large} items of size $\frac{2}{2k+1}+\varepsilon$ followed by  $2k+1$ \emph{small} items of size $\frac{1}{2k+1}$. We suppose that $ \frac{1}{2k+1} \gg \varepsilon >0$. The size of the largest item is $\frac{2}{2k+1}+\varepsilon$.
The instance can be interrupted at any moment by the adversary when the solution maintained by the algorithm is (asymptotically) $(1-\frac{1}{2k+1})$-approximate, or worse. 

The first $k$ items (which are identical) fit into the knapsack. If one of these items is not kept by the algorithm, then the competitive ratio is at most $1-\frac{1}{k}$ when $\varepsilon$ tends to 0. From now on we suppose that the $k$ first items of the instance belong to the solution. 

When the $(k+1)$-th item of the instance is disclosed, it does not fit into the solution. However, the value of an optimal solution \pjreplace{($k*(\frac{2}{2k+1}+\varepsilon$)}{$k(\frac{2}{2k+1}+\varepsilon)$} does not change. The optimal value changes only if the very last item is disclosed (i.e., it becomes 1 with $2k+1$ small items).

In between, the algorithm can 
$(1)$ replace one large item with one small item, or
$(2)$ buffer a small item,  or $(3)$  replace one large item with two small items. 

Note that the first action decreases the value of the current solution by $\frac{1}{2k+1}+\varepsilon$, giving a competitive ratio of $1-\frac{1}{2k}$ when $\varepsilon$ tends to zero. Therefore, the adversary can stop the instance if the first action is done. The second and third actions do not significantly change the value of the current solution (no decrease or a decrease by $\varepsilon$). However, the third action can be performed at most  $k-1$ times because each of them requires a recourse. 

Since the instance comprises $k$ large items, the algorithm is forced, either to do the first action, or to reach the very end of the instance where the value of the optimum becomes 1 and the algorithm's solution has value at most $k(\frac{2}{2k+1}+\varepsilon)= 1-\frac{1}{2k+1} +k\varepsilon$. In any case, 
the competitive ratio is at best $1-\frac{1}{2k+1}$ when $\varepsilon$ goes to zero. \qed
\end{proof}

\noindent \textbf{Lemma  \ref{lemmaub1rec1}.} \emph{The competitive ratio of any algorithm with at most one recourse operation is upper-bounded by $\sqrt{\frac{9-\sqrt{17}}{8}} \approx 0.78$ when $u\ge \frac{9-\sqrt{17}}{8} \approx 0.61$.}
\begin{proof} Let $s=\frac{9-\sqrt{17}}{8} \approx 0.61$.  
The first items $o_1$ and $o_2$ have size $1-s+\varepsilon$ and $s$, respectively.\footnote{We suppose that $\varepsilon>0$ and $s \gg \varepsilon$ hold.} If  $o_2$ does not replace $o_1$, then the competitive ratio is  $\frac{1-s}{s}$ when $\varepsilon$ goes to $0$. From now on, suppose that $\{o_2\}$ is the current solution and $o_1$ is in the buffer. 
The third item $o_3$ has size $1-s+\varepsilon$.  
If the solution remains $\{o_2\}$, then we get a ratio of $\frac{s}{2(1-s+\varepsilon)}$ which tends to $\frac{s}{2(1-s)}=\sqrt{s}$ when $\varepsilon$ goes to $0$. Otherwise, one can use the recourse to produce the solution $\{o_1,o_3\}$. Then, a fourth item $o_4$ of size $1-s$ is disclosed. The optimum $\{o_2,o_4\}$ has size one, and the size of the current solution is at most $2(1-s+\varepsilon)$, so the competitive ratio \pjreplace{is $2(1-s)$ which tends to $\sqrt{s}$}{tends to $2(1-s)=\sqrt{s}$} when $\varepsilon$ goes to $0$.   
Given that $s>1-s+\varepsilon$ and $\frac{1-s}{s} < \sqrt{s}$, we get an upper bound of $\sqrt{s} \approx 0.78$ when $u \ge s$. \qed
\end{proof}
}

\subsection{Useful Details from \cite{GourvesP24}}
\label{ap:DetailsFromGP}

In this section we list further useful details from \cite{GourvesP24}. The competitive ratios are achieved by \nonadapt\, introduced in the same work.


\setcounter{algorithm}{4}
\begin{algorithm}
\caption{
\label{algo7}}
\begin{algorithmic}[1] 
\State $S \gets \emptyset$
\While{a new item $o_i$ arrives}
    \If{ \label{line1024} $v(S) \ge \gamma_k$}
        \State \label{line1025}  Reject $o_i$ \Comment{{\em $S$ is not changed afterwards}}
    \Else
\State \label{line1027} $S \gets S \cup \{o_i\}$

\If{\label{line1030} $v(S) > 1$}


\State \label{line1123} Let $B=\{o \in S : |o| > 1-\gamma_k \}$ 

\If{  $v(B) > 1$}

\If{  $B$ contains a subset $\hat S$ 
such that $1 \ge v(\hat S) \ge \gamma_k$}


\State \label{line291} $S \gets \hat S$ \Comment{{\em $S$ is not changed afterwards}} 

\Else

\State \label{line1044} Remove the largest item of $S$ 

\EndIf 

\Else

\While{$v(S)>1$ \label{line1050}}
\State \label{line1051} Remove from $S$ one element of $S \setminus B$ \Comment{{\em chosen arbitrarily}}
\EndWhile

\EndIf

\EndIf 
\EndIf 

\EndWhile

\end{algorithmic}
\end{algorithm}

%
%
%
%
%
%
%
%
%
%
%
%
%
%
%
%
%
%
%
%
%
%

The competitive 
ratios given in \cite{GourvesP24} (and the equivalent upper bounds), $\rho(u)$ of deterministic online algorithms for the unweighted knapsack problem with removable items are depicted in Figure \ref{fig_old}. The lower (solid)  and upper (dotted) bounds on $\rho(u)$ are non-increasing functions of $u$. One can observe that the competitive ratio is always above $t$\pjedit{, where $t=\phi^{-1}$}, and the performance guarantee increases (and tends to 1) as the size of the items reduces.

Along with Theorem~\ref{thm11} mentioned in Section~\ref{sec:adaptive}, the following theorem complements the result and corresponds to the ``sloped'' segments of the \cra\ figure,  providing an (almost tight) approximation guarantee for all possible values of $u$.

\begin{theorem} \label{thm111} For all integers $k \ge 2$, 
Algorithm \ref{algo7} is  $\frac{1-u}{(k-1)u}$-competitive when $u \in (\frac{1-\gamma_k}{\gamma_k^2}, \frac{1}{k-1}]$.
\end{theorem}

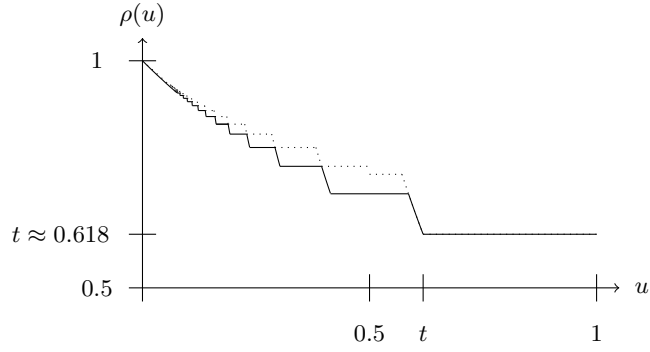
\begin{figure}
\center
\begin{tikzpicture}[scale=6.0]

\draw[->] (-0.03,0.5)--(1.05,0.5);
\draw[->] (0,0.47)--(0,1.05);


\draw  (1,0.47)--(1,0.53);
\draw  (0.03,1)--(-0.03,1);
\draw (-0.03,0.618033989)--(0.03,0.618033989);
\draw (0.618033989,0.53)--(0.618033989,0.47);
\draw (0.5,0.53)--(0.5,0.47);



\node () at (1.1,0.5) {$u$};
\node () at (0,1.1) {$\rho(u)$};
\node () at (-0.1,1) {$1$};
\node () at (0.618033989,0.4) {$t$};
\node () at (1,0.4) {$1$};
\node () at (-0.18,0.618033989) {$t\approx 0.618$};
\node () at (0.5,0.4) {$0.5$};
\node () at (-0.1,0.5) {$0.5$};

\draw[dotted] (0.618033989,0.618033989)--(1,0.618033989);
\draw[dotted] (0.394448725,0.767591879)--(0.5,0.767591879);
\draw[dotted] (0.5,0.75)--(4/7,0.75);
\draw[dotted] plot [domain=0.381966011:0.394448725,variable=\t] (\t,{(1-\t)/(2*\t)});
\draw[dotted] (0.291796068,0.809016994)--(0.381966011,0.809016994);
\draw[dotted] plot [domain=0.284450558:0.291796068,variable=\t] (\t,{(1-\t)/(3*\t)});
\draw[dotted] (0.229670386,0.838516481)--(0.284450558,0.838516481);
\draw[dotted] plot [domain=0.225148227:0.229670386,variable=\t] (\t,{(1-\t)/(4*\t)});
\draw[dotted] (0.188611699,0.86037961)--(0.225148227,0.86037961);
\draw[dotted] plot [domain=0.18567504:0.188611699,variable=\t] (\t,{(1-\t)/(5*\t)});
\draw[dotted] (0.159670332,0.877150706)--(0.18567504,0.877150706);
\draw[dotted] plot [domain=0.157670781:0.159670332,variable=\t] (\t,{(1-\t)/(6*\t)});
\draw[dotted] (0.138260621,0.890388203)--(0.157670781,0.890388203);
\draw[dotted] plot [domain=0.136843822:0.138260621,variable=\t] (\t,{(1-\t)/(7*\t)});
\draw[dotted] (0.121822171,0.901085803)--(0.136843822,0.901085803);
\draw[dotted] plot [domain=0.120784389:0.121822171,variable=\t] (\t,{(1-\t)/(8*\t)});
\draw[dotted] (0.108824378,0.909901951)--(0.120784389,0.909901951);
\draw[dotted] plot [domain=0.108042773:0.108824378,variable=\t] (\t,{(1-\t)/(9*\t)});
\draw[dotted] (0.098300563,0.917288177)--(0.108042773,0.917288177);
\draw[dotted] plot [domain=0.097697891:0.098300563,variable=\t] (\t,{(1-\t)/(10*\t)});
\draw[dotted] (0.089612167,0.923563544)--(0.097697891,0.923563544);
\draw[dotted] plot [domain=0.089138045:0.089612167,variable=\t] (\t,{(1-\t)/(11*\t)});
\draw[dotted] (0.082321372,0.928959478)--(0.089138045,0.928959478);
\draw[dotted] plot [domain=0.081941876:0.082321372,variable=\t] (\t,{(1-\t)/(12*\t)});
\draw[dotted] (0.076118446,0.933647701)--(0.081941876,0.933647701);
\draw[dotted] plot [domain=0.075810088:0.076118446,variable=\t] (\t,{(1-\t)/(13*\t)});
\draw[dotted] (0.070778347,0.937758198)--(0.075810088,0.937758198);
\draw[dotted] plot [domain=0.07052447:0.070778347,variable=\t] (\t,{(1-\t)/(14*\t)});
\draw[dotted] (0.066133776,0.941391109)--(0.07052447,0.941391109);
\draw[dotted] plot [domain=0.065922308:0.066133776,variable=\t] (\t,{(1-\t)/(15*\t)});
\draw[dotted] (0.062057851,0.944624787)--(0.065922308,0.944624787);
\draw[dotted] plot [domain=0.061879877:0.062057851,variable=\t] (\t,{(1-\t)/(16*\t)});
\draw[dotted] (0,1)--(0.058452652,0.947521397)--(0.061879877,0.947521397);
\draw[dotted] plot [domain=4/7:0.618033989,variable=\t] (\t,{(1-\t)/\t});

\draw plot [domain=0:0.076118446,variable=\t] (\t,{1/(1+\t)});

\draw plot [domain=0.585786438:0.618033989,variable=\t] (\t,{(1-\t)/\t});
\draw plot [domain=0.394448725:0.414213562,variable=\t] (\t,{(1-\t)/(2*\t)});
\draw plot [domain=0.291796068:0.302775638,variable=\t] (\t,{(1-\t)/(3*\t)});
\draw plot [domain=0.229670386:0.236067977,variable=\t] (\t,{(1-\t)/(4*\t)});
\draw plot [domain=0.192582404:0.188611699,variable=\t] (\t,{(1-\t)/(5*\t)});
\draw plot [domain=0.16227766:0.159670332,variable=\t] (\t,{(1-\t)/(6*\t)});
\draw plot [domain=0.140054945:0.138260621,variable=\t] (\t,{(1-\t)/(7*\t)});
\draw plot [domain=0.123105626:0.121822171,variable=\t] (\t,{(1-\t)/(8*\t)});
\draw plot [domain=0.108824378:0.109772229,variable=\t] (\t,{(1-\t)/(9*\t)});
\draw plot [domain=0.099019514:0.098300563,variable=\t] (\t,{(1-\t)/(10*\t)});
\draw plot [domain=0.090169944:0.089612167,variable=\t] (\t,{(1-\t)/(11*\t)});
\draw plot [domain=0.08276253:0.082321372,variable=\t] (\t,{(1-\t)/(12*\t)});
\draw plot [domain=0.076473219:0.076118446,variable=\t] (\t,{(1-\t)/(13*\t)});


\draw (0.618033989,0.618033989)--(1,0.618033989);
\draw (0.414213562,0.707106781)--(0.585786438,0.707106781);
\draw (0.394448725,0.767591879)--(0.302775638,0.767591879);
\draw (0.291796068,0.809016994)--(0.236067977,0.809016994);
\draw (0.229670386,0.838516481)--(0.192582404,0.838516481);
\draw (0.188611699,0.86037961)--(0.16227766,0.86037961);
\draw (0.188611699,0.86037961)--(0.16227766,0.86037961);

\draw (0.159670332,0.877150706)--(0.140054945,0.877150706);
\draw (0.138260621,0.890388203)--(0.123105626,0.890388203);
\draw (0.121822171,0.901085803)--(0.109772229,0.901085803);
\draw (0.108824378,0.909901951)--(0.099019514,0.909901951);
\draw (0.098300563,0.917288177)--(0.090169944,0.917288177);
\draw (0.089612167,0.923563544)--(0.08276253,0.923563544);
\draw (0.082321372,0.928959478)--(0.076473219,0.928959478);
\draw (0.076118446,0.933647701)--(0.071067812,0.933647701);
\end{tikzpicture}
\caption{Lower (solid) and upper (dotted) bounds on the competitive ratio $\rho$ as a function of $u$.\label{fig_old}}
\end{figure}

  The rationale of $\gamma_k$ originates from the following equality: 
  
\begin{equation}\label{gammakeq}
(k-1)(1-\gamma_k)+\frac{1-\gamma_k}{\gamma_k}=\gamma_k
\end{equation}
Other useful technical properties of $\gamma_k$ (valid for all positive integers $k$\pjedit{, except for~\eqref{gammakeq5}, which requires $k\ge2$}) are:

\begin{align}
0 < 1-\gamma_k \leq \frac{1-\gamma_k}{\gamma_k} &< 1/k \le \frac{1-\gamma_k}{\gamma_k^2}
\label{gammakeq3} \\
(k-1) \frac{1-\gamma_k}{\gamma_k^2} &\le 1
\label{gammakeq4} \\
(k+1)(1-\gamma_k) &\ge \gamma_k
\label{gammakeq2} \\
\frac{1-\frac{1-\gamma_k}{\gamma_k^2}}{(k-1)\frac{1-\gamma_k}{\gamma_k^2}} &= \gamma_k
\label{gammakeq5}
\end{align}

\subsection{\pjreplace{Omitted proofs from the Adaptive Algorithm with Removal Analysis}{Proofs from the Adaptive Algorithm with Removal Analysis}}
\label{ap:omittedAdaptiveRemoval}

\noindent \textbf{Observation~\ref{obs:optimal_k}. }
    Consider that the item upper bound is $u$. Then the optimal parameter $k^*$ for \nonadapt\ to guarantee the best possible approximation is: 
    \[
    k^* =  \max\Big\{k\in \mathbb{N}\text{ } \big | \text{ } \frac{2}{k-1+\sqrt{(k-1)^2+4}}\geq u \Big\}
    \]

\begin{proof}
    The observation easily follows from Theorems~\ref{thm11} and~\ref{thm111}.
    Specifically, as mentioned in \cite{GourvesP24}, $\frac{2}{k-1+\sqrt{(k-1)^2+4}}$ is the solution of the equation $\frac{1-u}{(k-1)u}=\gamma_{k-1}$\pjedit{ for $k\ge2$}. I.e., where the approximation guarantee given by Theorem~\ref{thm111} for $k$ is equal to the approximation guarantee given by Theorem~\ref{thm11} for $k-1$, making $k-1$ a better parameter for any larger $u'>u$ \pjedit{with $u'\leq \frac{1}{k-1}$}.

    Furthermore, for $k' = k+1$ and any \pjreplace{$u>\frac{2}{k'-1+\sqrt{(k'-1)^2+4}}$}{$\frac{2}{k'-1+\sqrt{(k'-1)^2+4}}<u\leq \frac{1}{k'-1}=\frac{1}{k}$}, Theorem~\ref{thm11} guarantees a better or equal approximation ratio using parameter $k$, i.e., $\gamma_k$. That is since $\frac{2}{k'-1+\sqrt{(k'-1)^2+4}}$ is the solution to $\frac{1-u}{(k'-1)u}=\gamma_{k'-1} = \gamma_k$ by substitution. \pjedit{For bounds above the respective endpoints, Theorem~\ref{thm11} provides no guarantee for the larger parameter in each comparison.}

    Finally, notice that for $k=1$, the ratio becomes $1$ (i.e., larger or equal to any $u$). \qed 
\end{proof}

\pjdelete{
\noindent \textbf{Claim~\ref{cl:noBigItemsDeleted}:}
    As long as the largest item size, $MaxItem$, seen so far by the algorithm satisfies  $MaxItem\leq \frac{2}{4-1+\sqrt{(4-1)^2+4}}$ (i.e., $k\geq 4$), Algorithm~\ref{algo_adaptive} has never removed an item from set $B$, i.e., a non-small item.\footnote{That is, no item is deleted while being a member of $B$ i.e., considered non-small at the current phase. An item that has been a member of $B$ but is in a later phase re-classified as small (i.e., member of $S$ but not member  of $B$) can get deleted.}

\begin{proof}
    We prove the claim by showing that because we are checking for a set $S$ satisfying an approximation ratio of $\gamma_{k'}$, instead of $\gamma_{k_{curr}}$\footnote{Where $k_{curr}$ the optimal $k$ for the items bounded by the highest value seen so far.}, the sum of items in $B$ will always become larger than $\gamma_{k'}$ before becoming infeasible (i.e., larger than 1). We show this by using Lemmas \ref{lm:sufficientNonSmallItems} and \ref{lm:largeItemsBound}.
    
        Specifically, by definition, at any point, the algorithm has seen items of size up to \[
    \frac{2}{k_{curr}-1+\sqrt{(k_{curr}-1)^2+4}}.
    \]
    
    By step~\ref{algStep:smallRemoval} of the algorithm\pjreplace{ - removing all non-small items from $B$ whenever $k'$ is updated}{, resetting $B$ at an overflow}, all items in $B$ \pjreplace{at any point}{at that point} must be of size \pjreplace{at least}{strictly greater than} $1-\gamma_{k'} = 1-\gamma_{k_{curr}-2}$.
    
    By Lemma~\ref{lm:sufficientNonSmallItems} we have that 
    
    \[
    (k-1)(1-\gamma_{k-2})\geq \gamma_{k-2}\Rightarrow (k_{curr}-1)(1-\gamma_{k_{curr}-2})\pjreplace{>}{\ge}\gamma_{k_{curr}-2}=\gamma_{k'}
    \].

This means that whenever there are $k_{curr}-1$ items in $B$, there is a set $S\subseteq B$ that satisfies \pjreplace{$|S|\geq \gamma_{k'}$}{$v(S)>\gamma_{k'}$}.

Finally, by Lemma \ref{lm:largeItemsBound}, we have: 
\[
(k_{curr}-1)\frac{2}{k_{curr}-1+\sqrt{(k_{curr}-1)^2+4}}\leq 1.
\]

The right part of the product on the left side of the inequality is an upper bound on the largest possible item size seen so far. Thus, whenever there are $k_{curr}-1$ items in $B$, they can fit in the knapsack and would be returned as the algorithm's solution. \qed
\end{proof}
}

\begin{lemma}
    \label{lm:sufficientNonSmallItems}
    For $k\geq3$,
    \[
    (k-1)(1-\gamma_{k-2})\geq \gamma_{k-1}\geq\gamma_{k-2}
    \]
\end{lemma}

\begin{proof}
It suffices to prove the first inequality, the second stands by definition of $\gamma_k$ as an increasing function in \cite{GourvesP24}.

We prove this by using the stronger Lemma~\ref{lm:itemStructure}, described and proven in Appendix~\ref{ap:adaptRec}:
For all integers $k\geq 3$ and $0\leq n < k-2$,
    \[
    1>(k-n-1)(1-\gamma_{k-2})+n(1-\gamma_{k-1})\geq \gamma_{k-1}
    \]
    
Specifically, all we need to note is that $1-\gamma_{k-2}> 1-\gamma_{k-1}$, since $\gamma_{k-1}>\gamma_{k-2}$. Thus, applying the lemma for $n=k-3$ we get:

\begin{align}
    2(1-\gamma_{k-2})+(k-3)(1-\gamma_{k-1})&\geq \gamma_{k-1}\Rightarrow\\
    (k-1)(1-\gamma_{k-2})&\geq \gamma_{k-1}
\end{align}

Which is the desired inequality. \qed

\end{proof}


\begin{lemma}
\label{lm:largeItemsBound}
For $k\geq 1$, 
\[
(k-1)\frac{2}{k-1 + \sqrt{(k-1)^2+4}}\leq 1.
\]
\end{lemma}

\begin{proof}
    For $k=1$, the left hand side is $0$ so the inequality holds.

    For $k>1$, we first notice that the denominator, on the left side, $k-1 + \sqrt{(k-1)^2+4}$, is positive. Thus the inequality is equivalent to:

    \begin{align}
        2(k-1)\leq{k-1 + \sqrt{(k-1)^2+4}} \Rightarrow\\
        (k-1)\leq  \sqrt{(k-1)^2+4}
    \end{align}

Which stands for all $k>1$. \qed
\end{proof}

\subsection{Proof of key Lemmata for the analysis of \adaptRec .}
\label{ap:adaptRec}

\pjdelete{

We begin by providing the formal description of algorithm \adaptRec:

Similarly to  the previous case of Algorithm~\ref{algo_adaptive}, our strategy is to show that our adaptive algorithm will behave \emph{almost} exactly, in this case, as the non adaptive algorithm would have if it was running for \pjreplace{$\hat{k}+1$}{$\tilde{k}$} from the beginning. Then, should an approximation-wise noticeable difference have emerged, we can identify it and use the recourse to ``reverse'' it to the required degree such that a solution with the targeted approximation is provided. Critically, we can always do this before the input ends.

We begin the analysis of the algorithm by stating some key observations and lemmata.

For the analysis, apart from the previous definition of the set of non-small items, $B$ (and its instance $B_i$ at step $i$), we also need to define an equivalent set of ``non-\emph{very}-small-items". That is, items of size $>\pjreplace{1-\gamma_{k'+1}}{1-\gamma_{k''}}$ (i.e., $k-1$-non-small items in the general case). We equivalently use the terms $C$ and $C_i$ for  this set. By definition, $B_i\subseteq C_i$ at any $i$. Even though Algorithm~\ref{algo_adaptive_recourse} still defines the set of non-small items to generally be $k-2$-non-small items, in this case we want to compare its behavior to the non-adaptive algorithm, Algorithm~\ref{algo7}, run with parameter $k-1$ (since we are trying to guarantee $\gamma_{k-1}$ approximation in this case).
\pjedit{Here $k=k_{curr}$ within the phase under consideration, and $B_i$ and
$C_i$ refer to the corresponding subsets of the currently held items.}

\begin{observation}
\label{obs:noDeletionofNonSmall}
    By Lemma~\ref{lm:sufficientNonSmallItems}, we have that the main condition used in the analysis of Algorithm~\ref{algo_adaptive}, i.e., that no item in the current set of non-small items $B_i$ will be deleted, also stands for Algorithm~\ref{algo_adaptive_recourse} \pjedit{before reaching the target, when $k_{curr}\geq3$}, where we are trying to reach $\gamma_{k-1}$ approximation.
\end{observation}

The proof is identical to Claim~\ref{cl:noBigItemsDeleted}. All we have to notice is that in that specific proof we didn't use  Lemma~\ref{lm:sufficientNonSmallItems} tightly: We only needed that  $(k-1)(1-\gamma_{k-2})\geq \gamma_{k-2}$ when it is also $(k-1)(1-\gamma_{k-2})\geq \gamma_{k-1}$.
\pjedit{The additional largest-item deletion in line~18 is performed only when
$B=\varnothing$, so it does not remove a non-small item.}

The following lemma is key to understanding the structure of non-small items. To maintain the narrative, we state the Lemma here and provide its proof after the proof of the main theorem, towards the end of the section:

\begin{lemma}
    \label{lm:itemStructure}
    For all integers $k\geq 3$ and $0\leq n < k-2$,
    \[
    1>(k-n-1)(1-\gamma_{k-2})+n(1-\gamma_{k-1})\geq \gamma_{k-1}.
    \]
\end{lemma}

\begin{corollary}
\label{col:itemsCombination}
    Lemma~\ref{lm:itemStructure} in combination with Lemma~\ref{lm:largeItemsBound} show that whenever there are $k-1$ non-very-small-items such that at least two of them are non-small, then their sum is a feasible solution of at least $\gamma_{k-1}$ approximation ratio.
\end{corollary}

\pjmargincomment{Simplifying/removed wrong invariant observation.}
\pjreplace{The following observation also gives us key insight on the structure of non-small items. Specifically, it states that there is always at least one of them, during the whole execution of the algorithm.

\begin{observation}
\label{obs:nonEmpty}
    For any $i$, $B_i$ is non empty (i.e., there is at least one non-small item at any step).
\end{observation}

\begin{proof}
    Remember that at each iteration we set $k_{curr}\gets \max\Big\{k\in \mathbb{N}\text{ } \big | \text{ } \frac{2}{k-1+\sqrt{(k-1)^2+4}}\geq MaxItem \Big\}$, where $MaxItem$ the maximum item size seen so far. That means that $MaxItem > \frac{2}{(k+1)-1+\sqrt{((k+1)-1)^2+4}}\geq 1-\gamma_{k'}$, for any $k\geq 2$ \pamargincomment{Do we need to add proof for this?}. This, in combination with Observation~\ref{obs:noDeletionofNonSmall}, prove  the Observation. \qed
\end{proof}

Corollary~\ref{col:itemsCombination} and Observation~\ref{obs:nonEmpty} give us a good picture of the items that have been included in the knapsack before the algorithm deletes any non-very-small items. 

Specifically, by Lemma~\ref{lm:largeItemsBound} we know that $k-1$ items of the maximum item size seen so far fit in the solution. By the fact that Algorithm~\ref{algo_adaptive_recourse} deletes non-small items in increasing size order, we know that if we have seen up to $k-1$ non-very-small items so far, none of them will have been deleted. 

Thus, by Corollary~\ref{col:itemsCombination}, we have that when we actually accumulate $k-1$ items we have one of two cases: Either we have at least 2 non-small items and the sum of non-very-small items immediately gives the approximation guarantee or we have exactly 1 non-small-item. 

For the latter case, it is not guaranteed that the sum of non-very-small in the solution gives approximation $\gamma_{k-1}$. Specifically, if we take all the non-small-items in this structure to be exactly on their lower bounds (i.e., 1 non-small item and $k-2$ non-very-small-items) we may have:

\[
1-\gamma_{k-2} + (k-2)(1-\gamma_{k-1})<\gamma_{k-1}
\]

That being said, the above only says that there is a small gap between the current sum and the desired approximation. Different item sizes forming the same structure (only 1 non-small and $k-2$ non-very-small) may still additively guarantee the approximation ratio.

}{\begin{observation}\label{obs:nonEmpty}
If Algorithm~\ref{algo_adaptive_recourse} deletes an item while $k_{\mathrm{curr}}=k\ge3$ and has not reached
its target, its solution after that iteration has value greater than
$\gamma_{k-2}$. If $k_{\mathrm{curr}}$ subsequently decreases, the new target is
already satisfied.
\end{observation}
\begin{proof}
By Observation~\ref{obs:noDeletionofNonSmall}, all deleted items have size at most
$1-\gamma_{k-2}$. Immediately before the last deletion that restores feasibility,
the solution has value greater than $1$. Therefore, afterwards its value is
greater than $\gamma_{k-2}$. Further iterations with the same parameter preserve
this bound. If the parameter decreases to $j\le k-1$, the new target is
$\gamma_{\max\{j-1,1\}}\le\gamma_{k-2}$, so line~9 applies before processing the
new item.
\end{proof}}

The following lemma completes the picture of the possible structures of non-small-items that could provide the approximation ratio:

\begin{lemma}
    \label{lm:nonVerSmallBound}
    For any $k\geq 2$,
    \[
    1>k(1-\gamma_{k-1})\geq \gamma_{k-1}
    \]
\end{lemma}

That is, if $k$ non-very-small items fit in the knapsack, then their sum gives the required approximation. 


\pjmargincomment{For updated proof}\pjreplace{The proof of the main theorem then is based in the following logic: Based on the observations above, realize that the only case in which Algorithm~\ref{algo_adaptive_recourse} may delete non-very-small items in a different order than the non adaptive algorithm would have before termination, can only be a case in which we have $k-1$ non-very-small items and exactly 1 non-small item.

The difference in behavior between the two algorithms may only be important if termination occurs or should have occurred in the same ``phase" (i.e., before $k_{curr}$ is updated): Whenever $k_{curr}$ is updated non-very-small items may only include items that were considered non-small in the previous phase - that never get deleted by Observation~\ref{obs:noDeletionofNonSmall}. Any non-very-small items of the previous phase that weren't non-small would be considered small in the new phase so their deletion order wouldn't matter. [!!!Rephrase/only include in the proof???!!!]

Thus, we may focus only on the possible termination conditions within each distinct phase in this case. These conditions can only be one of 3: a) The input ends. b) A combination of $k-1$ non-very-small items will appear in the input, s.t. their sum would suffice as a termination condition (i.e., guarantee the approximation ratio). c) The above is achieved by $k$ non-very-small items.\pamargincomment{Removed big paragraph from below that seemed redundant. Double check.}

}{The proof of the main theorem distinguishes the case where $B$ is empty
from the case where there is exactly one non-small item. In the former case,
Algorithm~\ref{algo_adaptive_recourse} removes the largest item, as Algorithm~\ref{algo7} would do. In the latter
case, we compare the values of the two solutions and show that the next
non-very-small item either gives a target solution directly or allows the use
of one recourse. By Observation~\ref{obs:nonEmpty}, a change in
$k_{\mathrm{curr}}$ before termination can occur only before any deletion.}

We are now ready to conclude with the formal proof of Theorem~\ref{Thm:Adaptive_recourse}:\pacomment{Probably tone done the pre-proof description or include it all in the proof? OR keep only this if it is informative enough and put formal proof in appendix???}
\newline

\noindent \textbf{Theorem~\ref{Thm:Adaptive_recourse}.}
        For all positive integers $k$, Algorithm~\ref{algo_adaptive_recourse} is $\gamma_{\tilde{k}}$-competitive when $u\leq \frac{2}{k-1+\sqrt{(k-1)^2+4}}$, where $\tilde{k}= \max\{k-1, 1\}$\pjedit{, using at most one recourse}.

\pjmargincomment{Core of updated/corrected proof -- validate}\pjreplace{\begin{proof}
    For bounding the approximation, the same argument we used in the proof of Theorem~\ref{thm:adaptiveMain} stands for $\tilde{k}$: Algorithm~\ref{algo7} (non-adaptive) is guaranteed to achieve approximation $\gamma_{\tilde{k}}$ whenever $u\leq \frac{2}{k-1+\sqrt{(k-1)^2+4}}$, by using the same arguments and Theorem~\ref{thm11}.

    Thus, following the logic of of the proof of Theorem~\ref{thm:adaptiveMain}, we have to show that Algorithm~\ref{algo_adaptive_recourse} behaves exactly the same \emph{or} approximation-wise equivalently to Algorithm~\ref{algo7} when run with parameter $\tilde{k}$. Specifically, in this case we want to show this equivalence focusing on sets $C_i$ (set of non-very-small items of adaptive algorithm) and $B_{k'',i}$ (set of non-small items of non adaptive algorithm when run with parameter $k''$)\pamargincomment{Does this need more explanation???}.

Our claim will hold for the $\gamma_{\tilde{k}}$ approximation for $\tilde{k}=\max\{k-1,1\}$ for any positive integer $k$. Specifically, we notice that all the lemmata mentioned work for $k\geq 3$. Thus, as $k_{curr}$ decreases they apply to all $\tilde{k}\geq 2$. Finally, we notice that when $k_{curr} \in \{1,2\}$ then $\tilde{k}=1$ and $k'=k''=1$. In that case our algorithm reverts to the basic case of Algorithm~\ref{algo7} run for parameter 1, guaranteeing approximation $\gamma_1$ for any item size. \footnote{Notice, that by Lemma~\ref{lemma:k-sizes}, since at the iteration, $i$, right before updating $k''$ to 1 we had $k_{curr}\geq 3$, and no items  in $B_i$ have been deleted by Observation~\ref{obs:noDeletionofNonSmall}, no 1-non-small item will have been deleted at the start of the final phase (thus the set of non-small items will be identical to the set of 1-non-small items of the non adaptive algorithm at that step, $B_{i,1}$). }

We begin by identifying the only case in which it could be $C_i\neq B_{k'',i}$ for any $i$: Since by Observation~\ref{obs:noDeletionofNonSmall}, no item in $B_{i-1}$ has been deleted until the end of step $i-1$, whenever $k''$ is updated from $k''_{i-1}$ to $k''_i$ ($k''_{i-1}> k''_{i}$), all items that have been deleted so far are small w.r.t. $k''_{i}$, and thus their deletion order doesn't matter. Furthermore, it will be $C_i=B_{k''_{i},i}$ since $B_{k''_{i},i}$ either only contains items from $B_{i-1}$ when ($k''_{i}= k''_{i-1}-1$), or only contains the new item that arrived.

Furthermore, since $k-1$ items of the largest size seen-so-far always fit in the knapsack by Lemma~\ref{lm:largeItemsBound}, non items in $C_i$ will be deleted if $|C_i|\leq{k-1}$, and only smaller items may have been deleted due to the algorithm deleting small items in increasing size.

Thus, we only need to focus in the case where we have $|C_i|\geq k-1$. By Observation~\ref{obs:nonEmpty} we know that in that case it is $|B_i|\geq 1$. Furthermore, if $|B_i|\geq 2$ then by Lemma~\ref{lm:itemStructure} we have that the items in $C_i$ satisfy the termination condition of the algorithm, achieving the desired approximation.

This leads us to the only case where Algorithm~\ref{algo_adaptive_recourse} may delete items from $C_i$ while Algorithm~\ref{algo7} would delete items from the so-far equivalent $B_{k''_i,i}$ in a different order: when we have exactly $k-1$ non-very-small items such that their sum doesn't satisfy the termination condition, and another non-small item arrives that doesn't fit in the knapsack. In that case, Algorithm~\ref{algo_adaptive_recourse} will delete the smallest item in $C_i$, while Algorithm~\ref{algo7} would delete the largest item in $B_{k''_i,i}$. To show the equivalence of the two algorithms in this case, we focus on the two possible termination conditions:

\begin{itemize}
    \item \textbf{Condition satisfaction:} In this case we want to show that if the items that appear on and after the phase $i$ are sufficient to satisfy the termination condition, then Algorithm~\ref{algo_adaptive} will identify it and manage to provide the equivalent solution, possibly using the recourse. We can focus only on non-very-small items, since any small items that arrive would be treated the same in both cases:

    By Lemma~\ref{lm:nonVerSmallBound} we have that if $k$ non-very-small items can fit in the knapsack, then they guarantee the termination condition. Thus, it could be the case that instead of deleting the smallest item in $C_i$, to accommodate the new item, we should have deleted the largest one in order for $k$ items to eventually fit in $C$.

    We first note the following inequality, that upper bounds the possible gap between the size of $k-1$ non-very-small-items, including at least one non-small, and the satisfaction of the condition (analysis of the inequality follows after the theorem's proof):

    \begin{equation}
    \label{eq:satisfactionSlack}
    \gamma_{k-1} - \big [(1 - \gamma_{k-2}) + (k - 2)(1 - \gamma_{k-1})\big ]< \frac{1-\gamma_{k-1}}{k+2}
    \end{equation}\pamargincomment{Should this be stated as separate lemma?}

    Furthermore we identify the gap between $k$ items of the lower bound of non-small-item size and the size of the knapsack (i.e., how much ``extra size'' other than the minimum provided by $k$ non-very-small items can fit in the knapsack. As with before, analysis of the inequality follows after the theorem's proof):

    \begin{equation}
        \label{eq:nonVerySmallBoundBrake}
        1-k(1-\gamma_{k-1})>\frac{1-\gamma_{k-1}}{k}
    \end{equation}\pamargincomment{Same question as for the previous equation}

    We are now ready to characterize the behavior of Algorithm~\ref{algo_adaptive_recourse} in this case: Since in phase $i$ we have $k-2$ non-very-small items and 1 non-small and the termination condition is not satisfied, by Equation~\ref{eq:satisfactionSlack} we have that the smallest item in $C_i$ is upper bounded by:
    
    \begin{equation}
        \label{eq:SmallestNonVrySmall}
        \min\{v\in C_i\}\leq (1-\gamma_{k-1}) + \frac{\frac{1-\gamma_{k-1}}{k+2}}{k-2}\leq (1-\gamma_{k-1})\Big (1+\frac{1}{k^2}\Big ),
    \end{equation}

    by pigeonhole principle.

    Furthermore, the sum of non-very-small items, excluding the non-small item, is upper bounded by the following, since the termination condition is not met:

    \begin{equation}
        \label{eq:LargestNonVrySmall}
        (1-\gamma_{k-1})\Big (k-2 +\frac{1}{k+2}\Big )
    \end{equation}

    As we move on from phase $i$, when the next non-very-small item appears, either the satisfaction condition will be met by fitting in the solution, either by deleting substituting (deleting) the current smallest item or it will delete the smallest item and the satisfaction condition will not be met. So far, not enough items have appeared to satisfy the termination condition.

    The important observation is this: If termination hasn't happened, all previously stated equations still stand for the same reason, for the equivalent set $C_{i'}$ ($i'>i$). We close the argument by showing that the next time that an item from $C_{i'}$ is about to be deleted due to another non-very-small item arriving that similarly doesn't satisfy the termination condition, we can bring back the previously deleted item, delete the non-small item and keep all the non-very-small items, including the new one, to reach a set that guarantees the termination condition. Specifically, in that case, the upper bound of the items in the knapsack, using Equation~\ref{eq:LargestNonVrySmall} and Equation~\ref{eq:SmallestNonVrySmall} will be:
    
    \begin{equation}
        \label{eq:bringBack}
        \sum {C_{i''}}\leq (1-\gamma_{k-1})\Big (k-2 +\frac{1}{k+2}\Big )+ 2 (1-\gamma_{k-1})\Big (1+\frac{1}{k^2}\Big ) < (1-\gamma_{k-1})\Big (k+\frac{1}{k}\Big)
    \end{equation}

    But, by Equation~\ref{eq:nonVerySmallBoundBrake}, we have that in that case, $C_{i''}$ is feasible (fits in the knapsack). Furthermore, $|C_{i''}|=k$ and thus, by Lemma~\ref{lm:nonVerSmallBound}, $C_{i''}$ satisfies the terminating condition.
    
    \item \textbf{Input ends:} If not enough items have arrived to reach the termination condition and the input ends, the current solution of Algorithm~\ref{algo_adaptive_recourse} guarantees at least as good as an approximation as Algorithm~\ref{algo_adaptive} since it has only deleted smaller items. This final case concludes the proof. \qed
\end{itemize}

\end{proof}}{\begin{proof}
Since $\gamma_k$ increases with $k$, it suffices to prove the guarantee for
$k=k^*$ of Observation~\ref{obs:optimal_k}. We first consider a phase during which
$k_{\mathrm{curr}}=k\ge3$. Write
\[
\beta=\gamma_{k-1},\qquad a=1-\gamma_{k-2},\qquad c=1-\beta.
\]
Thus $a>c$, $B=\{o\in S:|o|>a\}$ and $C=\{o\in S:|o|>c\}$.
Algorithm~\ref{algo7} run with parameter $k-1$ guarantees approximation $\beta$ for the
item-size upper bound of this phase, by Theorem~\ref{thm11} and the
same argument used in the proof of Theorem~\ref{thm:adaptiveMain}.

Suppose that an item causes $v(S)>1$, but the test in line~15 finds no subset
with value in $[\beta,1]$. Then $v(C)>1$. Otherwise, deleting items outside $C$,
each of size at most $c=1-\beta$, until feasibility is restored would leave a
subset of value greater than $\beta$, a contradiction. Just before the new item
arrived, the solution contained at most $k-1$ items in $C$: a feasible set of
$k$ such items would already reach the target by Lemma~\ref{lm:nonVerSmallBound}.
By Lemma~\ref{lm:largeItemsBound}, any $k-1$ items fit. Hence the current set $C$ has
exactly $k$ items. By Corollary~\ref{col:itemsCombination}, at most one of them belongs to
$B$, since otherwise a subset of $k-1$ items would reach the target.

If $B=\varnothing$, Algorithm~\ref{algo_adaptive_recourse} removes the largest item. This is exactly the
choice of Algorithm~\ref{algo7} when $v(C)>1$. It restores feasibility, since the removed
item is at least as large as the arriving item. Until the first deletion with
$B\ne\varnothing$, the two algorithms therefore have the same solution, unless
Algorithm~\ref{algo_adaptive_recourse} has already reached its target.

Consider that first deletion with $B\ne\varnothing$, and let $z>a$ be the size
of its unique non-small item. Let $d_0$ be the size of the smallest other item
in $C$, let $R$ be the total size of the remaining $k-2$ items in $C$, and let
$w$ be the total size of the items outside $C$. Algorithm~\ref{algo_adaptive_recourse} deletes all items
outside $C$ and then the item of size $d_0$, leaving value $z+R$; Algorithm~\ref{algo7}
removes the item of size $z$, leaving value $d_0+R+w$.
The former solution is feasible by Lemma~\ref{lm:largeItemsBound}. The latter is
feasible because $z$ is the largest item and the solution was feasible before
the arrival. Neither has value at least $\beta$, since both are subsets tested
in line~15.

In fact, the failed test also implies $z+R+w<\beta$. If this were false,
starting with the feasible set of value $z+R<\beta$ and adding items outside
$C$ until reaching $\beta$ would produce a subset of value at most
$\beta+c=1$. Since $d_0$ is the smallest of the $k-1$ items in $C\setminus B$,
we have $R-d_0\ge(k-3)c$. By Lemma~\ref{lm:itemStructure},
\[
2a+(k-3)c\ge\beta.
\]
Consequently,
\[
w<\beta-z-R\le z-d_0.
\]
Thus $z+R>d_0+R+w$: Algorithm~\ref{algo_adaptive_recourse} retains at least as much value as Algorithm~\ref{algo7}.
If the input ends now, its competitive ratio is at least $\beta$. Until another
item in $C$ arrives, items of size at most $c$ fit in both solutions while their
values are below $\beta$, and preserve this comparison.

Consider the next item of size greater than $c$. If a target subset exists,
Algorithm~\ref{algo_adaptive_recourse} finds it. Otherwise the preceding argument again gives exactly
$k$ items in $C$, with the same unique non-small item of size $z$. A second
non-small item would yield a target subset by Corollary~\ref{col:itemsCombination}.
Let $d_1$ be the smallest size among the other $k-1$ items in $C$, and let $R'$
be the sum of the remaining $k-2$ sizes. The set of value $z+R'$ is feasible
and has value less than $\beta$, so
\[
R'<\beta-z,\qquad d_1\le\frac{R'}{k-2}<\frac{\beta-z}{k-2}.
\]
The previous deletion similarly gives
$d_0\le R/(k-2)<(\beta-z)/(k-2)$. By restoring the item of size $d_0$ and
removing the non-small item, we obtain $k$ items in $C$ with total size
\begin{equation}\label{eq:recourse-corrected-sum}
R'+d_1+d_0<\frac{k(\beta-z)}{k-2}
<\frac{k(\beta-a)}{k-2}\le1.
\end{equation}
To verify the last inequality, Lemma~\ref{lm:itemStructure} gives
$2a+(k-2)c\ge1$. For $k\ge4$, using $a>c$, we obtain
\[
k(a+c)=2\bigl(2a+(k-2)c\bigr)+(k-4)(a-c)\ge2.
\]
Since $\beta-a=1-(a+c)$, the inequality follows. For $k=3$, it follows directly
from
\[
a+c=\frac{5-\sqrt5-\sqrt2}{2}>\frac23.
\]
By Lemma~\ref{lm:nonVerSmallBound}, the feasible set in~\eqref{eq:recourse-corrected-sum}
has value greater than $kc\ge\beta$. It is available in the test of line~20
by using the previously deleted item of size $d_0$. The algorithm therefore
uses one recourse and reaches its target. No further recourse is used.

It remains to consider changes in $k_{\mathrm{curr}}$. By
Observation~\ref{obs:nonEmpty}, a decrease after any deletion immediately
satisfies the new target. Every change before that has retained all items seen
so far, so the preceding comparison with Algorithm~\ref{algo7} begins from the same
solution. If $k_{\mathrm{curr}}\in\{1,2\}$, then $k'=k''=1$. After a failed target
subset test, $v(B)>1$ necessarily holds, by the same small-item deletion
argument. The algorithm then behaves as Algorithm~\ref{algo7} with parameter $1$ and
guarantees $\gamma_1$.

Finally, a solution that reaches its current target also meets every later
target, since $k_{\mathrm{curr}}$ can only decrease. If the input ends before
this happens, the comparison with Algorithm~\ref{algo7} in the final phase gives the
claimed ratio. This completes the proof.
\end{proof}}
}

Next, we provide the proof of Lemma~\ref{lm:itemStructure} (restated below):\\ \newline
\noindent \textbf{Lemma~\ref{lm:itemStructure}.}
    For all integers $k\geq 3$ and $0\leq n < k-2$,
    \[
    1>(k-n-1)(1-\gamma_{k-2})+n(1-\gamma_{k-1})\geq \gamma_{k-1}
    \]

\begin{proof}
    It suffices to notice that:
    \begin{align}
    2(1-\gamma_{k-2})+(k-3)(1-\gamma_{k-1})&\geq \gamma_{k-1}\Rightarrow \\
    2(1-\gamma_{k-2})+(k-3)(1-\gamma_{k-1}) - \gamma_{k-1}&\geq 0
    \label{eq:mainInequality}
    \end{align}
    That is, since $1-\gamma_{k-2}>1-\gamma_{k-1}$, the inequality becomes stronger as $n$ decreases. Furthermore, the upper bound of 1 follows from \pjreplace{Lemma~\ref{lm:largeItemsBound}, showing that $k-1$ items of the largest item size seen so far always fit in the knapsack.}{$1-\gamma_{k-2}=2/(k+\sqrt{(k-2)^2+4})<1/(k-1)$, since the expression contains $k-1$ terms, each at most $1-\gamma_{k-2}$.}
    
    We first rewrite the left-hand side:
\begin{align}
&2(1-\gamma_{k-2})+(k-3)(1-\gamma_{k-1})-\gamma_{k-1} \notag\\
&=2(1-\gamma_{k-2})+(k-2)(1-\gamma_{k-1})-1.
\end{align}

From the definition of $\gamma_k$, we obtain
\begin{align}
1-\gamma_k=\frac{k+2-\sqrt{k^2+4}}{2k}.
\end{align}

Therefore,
\begin{align}
&2(1-\gamma_{k-2})+(k-2)(1-\gamma_{k-1})-1 \notag\\
&=\frac{k-\sqrt{(k-2)^2+4}}{k-2}
+\frac{k-2}{2(k-1)}\bigl(k+1-\sqrt{(k-1)^2+4}\bigr)-1.
\end{align}

Multiplying by $2(k-2)(k-1)>0$, we see that \eqref{eq:mainInequality} is equivalent to
\begin{align}\label{eq:auxInequality}
N(k)\ge 0,
\end{align}
where
\begin{align}
N(k)= (k^3-3k^2+4k) -(k-2)^2\sqrt{(k-1)^2+4}
-2(k-1)\sqrt{(k-2)^2+4}.
\end{align}

For any $x>0$, we use the inequality
\begin{align}\label{eq:sqrtBound}
\sqrt{x^2+4}
=x\sqrt{1+\frac{4}{x^2}}
\le x\left(1+\frac12\cdot\frac{4}{x^2}\right)
=x+\frac{2}{x}.
\end{align}

Applying \eqref{eq:sqrtBound} with $x=k-1$ and $x=k-2$ (valid for $k\ge 3$), we obtain
\begin{align}
\sqrt{(k-1)^2+4}&\le (k-1)+\frac{2}{k-1},\\
\sqrt{(k-2)^2+4}&\le (k-2)+\frac{2}{k-2}.
\end{align}

Substituting these bounds into $N(k)$ gives
\begin{align}
N(k)\ge\;&k^3-3k^2+4k \notag\\
&-(k-2)^2\left((k-1)+\frac{2}{k-1}\right)
-2(k-1)\left((k-2)+\frac{2}{k-2}\right).
\end{align}

A direct simplification yields:\pamargincomment{Not adding more detail here, just a bunch of manipulations.}
\begin{align}\label{eq:lowerBoundNk}
N(k)\ge \frac{2(k^2-6k+6)}{(k-1)(k-2)}.
\end{align}

For $k\ge 3$, we have $(k-1)(k-2)>0$, and
\begin{align}
k^2-6k+6=(k-3)^2-3.
\end{align}

Thus, for all integers $k\ge 5$,
\begin{align}
k^2-6k+6\ge 0,
\end{align}
and therefore, by \eqref{eq:lowerBoundNk},
\begin{align}
N(k)\ge 0.
\end{align}

It remains to check the cases $k=3$ and $k=4$.

For $k=3$, inequality \eqref{eq:mainInequality} becomes
\begin{align}
2(1-\gamma_1)\ge \gamma_2.
\end{align}
We compute and get:
\begin{align}
3-\sqrt5\ge \frac{\sqrt2}{2},
\end{align}
which is true.

For $k=4$, inequality \eqref{eq:mainInequality} becomes
\begin{align}
2(1-\gamma_2)+(1-\gamma_3)\ge \gamma_3,
\end{align}
or equivalently,
\begin{align}
3-2\gamma_2\ge 2\gamma_3.
\end{align}
We compute and get:
\begin{align}
3-\sqrt2\ge \frac{1+\sqrt{13}}{3},
\end{align}
which is also true.

Therefore, inequality \eqref{eq:mainInequality} holds for $k=3,4$ and for all integers $k\ge 5$. Hence, it is valid for all integers $k\ge 3$. \qed
\end{proof}

We continue with the proof of Lemma~\ref{lm:nonVerSmallBound} .
\newline

\noindent\textbf{Lemma~\ref{lm:nonVerSmallBound}.}
    For any $k\geq 2$,
    \[
    1>k(1-\gamma_{k-1})\geq \gamma_{k-1}
    \]

\begin{proof}
Since $k\ge 2$ and set $n=k-1$. Then $n\ge 1$, and the inequality
\begin{align}\label{eq:mainInequality2}
1>k(1-\gamma_{k-1})\ge \gamma_{k-1}
\end{align}
is equivalent to
\begin{align}
1>(n+1)(1-\gamma_n)\ge \gamma_n.
\end{align}

Recall that
\begin{align}
1-\gamma_n&=\frac{n+2-\sqrt{n^2+4}}{2n} \Rightarrow\\
(n+1)(1-\gamma_n)
&=\frac{(n+1)\bigl(n+2-\sqrt{n^2+4}\bigr)}{2n}.
\label{eq:intermediate}
\end{align}

We first prove that $(n+1)(1-\gamma_n)<1$.

The inequality
\begin{align}
\frac{(n+1)\bigl(n+2-\sqrt{n^2+4}\bigr)}{2n}<1
\end{align}
is equivalent, since $2n>0$, to
\begin{align}
(n+1)\bigl(n+2-\sqrt{n^2+4}\bigr)&<2n\Rightarrow\\
n^2+n+2&<(n+1)\sqrt{n^2+4}.
\end{align}

Dividing by $n+1>0$, we obtain
\begin{align}
\frac{n^2+n+2}{n+1}&<\sqrt{n^2+4}\Rightarrow\\
n+\frac{2}{n+1}&<\sqrt{n^2+4}.
\end{align}

Both sides are positive, so we may square:
\begin{align}
\left(n+\frac{2}{n+1}\right)^2&<n^2+4\Rightarrow\\
\frac{n}{n+1}+\frac{1}{(n+1)^2}&<1\Rightarrow\\
\frac{n(n+1)+1}{(n+1)^2}&<1\Rightarrow\\
n^2+n+1&<n^2+2n+1.
\end{align}

Which reduces to $n>0$, which proves this direction as $n\ge 1$.

We now prove that $(n+1)(1-\gamma_n)\ge \gamma_n$. By \eqref{eq:intermediate} and the definition of $\gamma_k$ this is equivalent to:
\begin{align}
\frac{(n+1)\bigl(n+2-\sqrt{n^2+4}\bigr)}{2n}
\ge
\frac{n-2+\sqrt{n^2+4}}{2n}.
\end{align}

Since $2n>0$ the above becomes:
\begin{align}
(n+1)\bigl(n+2-\sqrt{n^2+4}\bigr)&\ge n-2+\sqrt{n^2+4}\Rightarrow\\
(n+1)(n+2)-(n-2)&\ge (n+2)\sqrt{n^2+4}\Rightarrow\\
n^2+2n+4&\ge (n+2)\sqrt{n^2+4}.
\end{align}

Dividing by $n+2>0$, we find
\begin{align}
n+\frac{4}{n+2}\ge \sqrt{n^2+4}.
\label{eq:secondIntermediate}
\end{align}

Using the inequality
\begin{align}
\sqrt{n^2+4}\le n+\frac{2}{n},
\end{align}
valid for $n>0$. Thus, for $n\ge 2$, we get:
\begin{align}
n+\frac{4}{n+2}\ge n+\frac{2}{n}\ge \sqrt{n^2+4},
\end{align}
verifying \eqref{eq:secondIntermediate}.

For the remaining case $n=1$ (that is, $k=2$), we evaluate both sides directly and get:

\begin{align}
3-\sqrt5\ge \frac{-1+\sqrt5}{2},
\end{align}

Which holds and completes the proof. \qed

\end{proof}

\pjdelete{We continue with the proof of \textbf{Equation~\ref{eq:satisfactionSlack}}, for every $k\geq 3$:

    \begin{equation*}
    \gamma_{k-1} - \big [(1 - \gamma_{k-2}) + (k - 2)(1 - \gamma_{k-1})\big ]< \frac{1-\gamma_{k-1}}{k+2}
    \end{equation*}
    
\begin{proof}

Manipulating the left part we get that \eqref{eq:satisfactionSlack} is equivalent to:

\begin{align}
\gamma_{k-2}-(k-1)(1-\gamma_{k-1})
&<
\frac{1-\gamma_{k-1}}{k+2}
\label{eq:rearranged2} \Rightarrow \\
\gamma_{k-2}
&<
(1-\gamma_{k-1})\left((k-1)+\frac{1}{k+2}\right)
\label{eq:rearranged3}\Rightarrow \\
\gamma_{k-2}
&<
(1-\gamma_{k-1})\frac{k^2+k-1}{k+2}.
\label{eq:keyInequality}
\end{align}

From the definition of $\gamma_k$,
\begin{align}
1-\gamma_k
&=1-\frac{k-2+\sqrt{k^2+4}}{2k} \nonumber\\
&=\frac{k+2-\sqrt{k^2+4}}{2k}.
\end{align}

Rationalizing the numerator, we obtain
\begin{align}
1-\gamma_k
=
\frac{2}{k+2+\sqrt{k^2+4}}.
\label{eq:oneMinusGamma}
\end{align}

Thus, we have:
\begin{align}
1-\gamma_{k-1}
=
\frac{2}{k+1+\sqrt{(k-1)^2+4}}.
\label{eq:oneMinusGammaKm1}
\end{align}

Also,
\begin{align}
\gamma_{k-2}
=
\frac{k-4+\sqrt{(k-2)^2+4}}{2(k-2)}.
\label{eq:gammaKm2}
\end{align}

Substituting \eqref{eq:oneMinusGammaKm1} and \eqref{eq:gammaKm2} into
\eqref{eq:keyInequality}, it is enough to show
\begin{align}
\frac{k-4+\sqrt{(k-2)^2+4}}{2(k-2)}
<
\frac{2}{k+1+\sqrt{(k-1)^2+4}}
\cdot
\frac{k^2+k-1}{k+2}.
\label{eq:reducedInequality}
\end{align}

For any $x>0$, we have
\begin{align}
\sqrt{x^2+4}
=
x\sqrt{1+\frac{4}{x^2}}
\le
x\left(1+\frac{2}{x^2}\right)
=
x+\frac{2}{x},
\label{eq:sqrtBound2}
\end{align}
since $\sqrt{1+u}\le 1+\frac{u}{2}$ for $u\ge0$.

For $k\ge5$, applying \eqref{eq:sqrtBound2} with $x=k-2$ and $x=k-1$ gives
\begin{align}
\sqrt{(k-2)^2+4}
&\le (k-2)+\frac{2}{k-2}, \nonumber\\
\sqrt{(k-1)^2+4}
&\le (k-1)+\frac{2}{k-1}.
\label{eq:sqrtBounds}
\end{align}

Substituting these bounds into \eqref{eq:reducedInequality} yields a stronger
inequality which simplifies to
\begin{align}
k^3-5k^2+5k-10 \ge 0.
\label{eq:cubic}
\end{align}

For $k\ge5$, we have
\begin{align}
k^3-5k^2+5k-10>0.
\end{align}

Indeed, at $k=5$ the value is $15>0$, and the derivative
\begin{align}
3k^2-10k+5>0 \quad \text{for } k\ge5
\end{align}
shows that the polynomial is increasing. Hence,
\eqref{eq:cubic} holds for all $k\ge5$, and therefore
\eqref{eq:satisfactionSlack} holds for all $k\ge5$.

It remains to check $k=3$ and $k=4$.

For $k=3$,
\begin{align}
\gamma_2=\frac{\sqrt8}{4},\qquad
\gamma_1=\frac{-1+\sqrt5}{2}.
\end{align}
Substituting into \eqref{eq:satisfactionSlack} verifies the inequality.

For $k=4$,
\begin{align}
\gamma_3=\frac{1+\sqrt{13}}{6},\qquad
\gamma_2=\frac{\sqrt2}{2}.
\end{align}
Again, direct substitution verifies \eqref{eq:satisfactionSlack}.

This completes the proof. \qed

\end{proof}

Finally, we provide the proof of \textbf{Equation~\ref{eq:nonVerySmallBoundBrake}}:
    \begin{equation*}
        1-k(1-\gamma_{k-1})>\frac{1-\gamma_{k-1}}{k}
    \end{equation*}
    
\begin{proof}

First, we simplify $1-\gamma_k$ as:

\begin{align}
1-\gamma_k
&=\frac{k+2-\sqrt{k^2+4}}{2k}
= \frac{(k+2)^2-(k^2+4)}{2k\bigl(k+2+\sqrt{k^2+4}\bigr)}
\\
&=\frac{4k}{2k\bigl(k+2+\sqrt{k^2+4}\bigr)}
=\frac{2}{k+2+\sqrt{k^2+4}}.
\label{eq:nonSmallSimplified2}
\end{align}

Thus:
\begin{align}
1-\gamma_{k-1}
=
\frac{2}{k+1+\sqrt{(k-1)^2+4}}.
\label{eq:oneMinusGammaKm1Second}
\end{align}

Let
\begin{align}
a:=1-\gamma_{k-1}>0.
\end{align}

Then \eqref{eq:nonVerySmallBoundBrake} is equivalent to
\begin{align}
1
>
ka+\frac{a}{k}
=
a\left(k+\frac{1}{k}\right).
\label{eq:step1b}
\end{align}

Dividing both sides of \eqref{eq:step1b} by $k+\frac{1}{k}>0$, we obtain
\begin{align}
a
<
\frac{1}{k+\frac{1}{k}}
=
\frac{k}{k^2+1}.
\label{eq:step2b}
\end{align}

Using \eqref{eq:oneMinusGammaKm1Second}, inequality \eqref{eq:step2b} will follow from
\begin{align}
\frac{2}{k+1+\sqrt{(k-1)^2+4}}
<
\frac{k}{k^2+1}.
\label{eq:step3}
\end{align}

Since all denominators in \eqref{eq:step3} are positive for $k\ge3$, it is equivalent to
\begin{align}
2(k^2+1)
&<
k\Bigl(k+1+\sqrt{(k-1)^2+4}\Bigr)
\label{eq:step4}\Rightarrow \\
2k+\frac{2}{k}
&<
k+1+\sqrt{(k-1)^2+4}.
\label{eq:step5}\Rightarrow \\
\sqrt{(k-1)^2+4}
&>
(k-1)+\frac{2}{k}.
\label{eq:step6}
\end{align}

Set $x=k-1>0$. Then \eqref{eq:step6} becomes
\begin{align}
\sqrt{x^2+4}
>
x+\frac{2}{x+1}.
\label{eq:step7}
\end{align}

Both sides of \eqref{eq:step7} are positive, so we may square:
\begin{align}
x^2+4
&>
\left(x+\frac{2}{x+1}\right)^2
\label{eq:step8}\Rightarrow \\
1
&>
\frac{x}{x+1}+\frac{1}{(x+1)^2}.
\label{eq:step11}
\end{align}

Thus, \eqref{eq:step11} is equivalent to
\begin{align}
1
>
\frac{x^2+x+1}{(x+1)^2}
\quad\Longleftrightarrow\quad
(x+1)^2>x^2+x+1.
\label{eq:step13}
\end{align}

Expanding gives
\begin{align}
x^2+2x+1>x^2+x+1,
\end{align}
which is equivalent to
\begin{align}
x>0.
\end{align}

Since $x=k-1>0$ for $k\ge3$, this holds.

This completes the proof. \qed

\end{proof}

}

\pjreplace{\subsection{Learning-Augmented Version}}{\subsection{Proof of the key lemma for the Learning-Augmented Algorithm}}
\label{ap:predictions}

\pjdelete{
We begin by describing the strategy of the augmented algorithm. Since we have already presented algorithms \adapt\ and Algorithm~\ref{algo7} of \cite{GourvesP24} and their similarities, we make a more informal presentation, that focuses on presenting the key aspect of how we capitalize on their behavior.


Next, we produce the arguments that prove Theorem~\ref{thm:augmented}:

\begin{proof}
\pjreplace{For the second case of the theorem, $k<4$, we notice that whenever}{Whenever} $k_p<4$, we only run \adapt\ and ignore the prediction. Thus, we guarantee approximation of $\gamma_{\hat{k}}$\pjedit{, where $\hat{k}=\max\{k^*-2,1\}$,} by Theorem~\ref{thm:adaptiveMain} which bounds both the robustness and the consistency of the algorithm.

We now focus on the more interesting case where $k_p\geq 4$. If $k_{opt}$ is the optimal $k$ for the real upper bound $u$ according to Observation~\ref{obs:optimal_k}, we break the analysis into the following cases:

\begin{description}
    \item[Case $k_{opt}> k_p$.] In this case the prediction is not accurate and thus we need to show the robustness bound. We identify two sub-cases:
    \begin{description}
        \item[\pjreplace{$k_{opt-2}> k_{p}$:}{$k_{opt}-2> k_{p}$:}] Then, the algorithm will terminate during the run of \adapt\ in step \ref{step:adapt} guaranteeing approximation \pjreplace{$\gamma_{\hat{k}}>\gamma_{k-3}$}{$\gamma_{\max\{k_{opt}-2,1\}}>\gamma_{k_{opt}-3}$} (Theorem~\ref{thm:adaptiveMain}).
        \item[\pjreplace{$k_{opt-2}\leq k_{p}$:}{$k_{opt}-2\leq k_{p}$:}] \pjreplace{Then the algorithm will terminate during the run of Algorithm~\ref{algo7} at step\ref{step:nonAdapt}.}{Then the algorithm reaches step~\ref{step:nonAdapt}, unless the required guarantee has already been achieved.} Notice that our algorithm will immediately switch to Algorithm~\ref{algo7} as soon as $k'\leq k_{p}$.\pjedit{ This switch is performed before processing the arriving item. Before this arrival, $k_{curr}\geq k_p+3$, unless this is the first item.} That is, at the point of the switch, all items \pjreplace{seen so far}{preceding the arriving item} would be considered small, w.r.t to $k_p$, by Lemma~\ref{lemma:k-sizes}. Thus, using the same logic as in the proof of Theorem~\ref{thm:adaptiveMain}, all \pjreplace{items seen so far}{earlier items} would have been considered as $k_p-small$ by Algorithm~\ref{algo7} and their deletion order doesn't matter. \pjreplace{That is, until this point, the algorithm has behaved exactly as if we were running Algorithm~\ref{algo7} from the beginning.}{More precisely, by the proof of Theorem~\ref{thm:adaptiveMain}, any earlier deletion or rejection leaves value at least $\gamma_{k_p}$, since until the switch $k'>k_p$. Otherwise, all earlier items remain. Thus, unless the value already reaches $\gamma_{k_p}$, the current state is the same as for Algorithm~\ref{algo7} with parameter $k_p$ from the beginning. The arriving item is then processed by Algorithm~\ref{algo7}.} Since $k_p<k_{opt}$ \pjreplace{(i.e., $u_p>u$)}{(hence $u_p>u$)}, Algorithm~\ref{algo7} will still provide approximation $\gamma_{k_p}$, by Theorem~\ref{thm11} \cite{GourvesP24}. But we have \pjreplace{$k_p\geq k_{opt-2}$}{$k_p\geq k_{opt}-2$} which is better than the desired robustness.
    \end{description}
    \item[Case $k_{opt}=k_p$.] \pjreplace{I.e., the prediction agrees with the real upper-bound and we need to show the consistency bound.}{First suppose that the prediction is never invalidated.} This case follows \pjreplace{exactly the same logic}{the same switching argument} as the sub case \pjreplace{($k_{opt-2}\leq k_{p}$)}{($k_{opt}-2\leq k_{p}$)} of the previous case. All we have to notice is that since now $k_{opt}=k_p$, Algorithm~\ref{algo7} will guarantee its best possible competitive ratio as stated by Theorem~\ref{thm11} \pareplace{and Theorem~\ref{thm111}}  (mentioned in Section~\ref{sec:adaptive}) of \cite{GourvesP24}.\pjedit{ This proves consistency when $u_p=u$, and robustness otherwise, since this guarantee is at least $\gamma_{k_p-1}\geq\gamma_{k_p-3}$. If the prediction is invalidated, the argument in the following case applies as well.}
\item[\pjreplace{Case $k_{opt}>k_p$.}{Case $k_{opt}<k_p$.}] I.e., the prediction is not accurate and thus we need to show the robustness bound. \pjreplace{The only case where the procedure of \adapt\ would be disrupted, is if at any step, $k' \in \{k_p, k_{p-1}, k_{p-2}\}$, thus entering step \ref{step:nonAdapt} of the algorithm.}{An item invalidating the prediction must arrive. The following argument also applies when $k_{opt}=k_p$ and the prediction is invalidated.} \pjreplace{But, if at any point after that, an item}{When the first item} invalidating the prediction arrives (i.e., an item of size $|o|>u_p$)\pjreplace{, and the current termination bound isn't reached, then we immediately return to making an over-estimation of $k$, by setting $k'< k_{p-3}$.}{, we immediately return to running \adapt\ with $k'=\max\{1,\min\{k_p-3,k_{curr}-2\}\}$ at every subsequent iteration.} By Lemma~\ref{lemma:k-sizes}, this means that \pjreplace{all previously seen items are now $k_{p-3}$-small}{all items preceding this arrival are $\{k_p-3\}$-small}, so their deletion order doesn't matter when we are targeting an approximation ratio of at most \pjreplace{$\gamma_{k_{p-3}}$}{$\gamma_{k_p-3}$}.\pjedit{ More precisely, if an earlier deletion has occurred, the last deletion restoring feasibility leaves value greater than $\gamma_{k_p-3}$. Selection of a feasible subset or rejection at a termination condition also leaves at least this value. This already suffices, since $k_{opt}\leq k_p$. Otherwise, all earlier items remain. The proof of Theorem~\ref{thm:adaptiveMain} applies to the subsequent run because $k'$ never increases and, whenever $k'>1$, satisfies $k'\leq k_{curr}-2$: by Lemmas~\ref{lm:sufficientNonSmallItems} and~\ref{lm:largeItemsBound}, $k'+1$ non-small items fit and already reach $\gamma_{k'}$, so no non-small item is removed before reaching the target; deleting small items to restore feasibility reaches the target. The case $k'=1$ is handled as in that proof. The final guarantee is therefore $\gamma_{\max\{1,\min\{k_p-3,k_{opt}-2\}\}}$, which is at least $\gamma_{k_{opt}-3}$ when $k_{opt}\geq4$, and at least $\gamma_1$ otherwise.} Thus the desired guarantee directly follows from the analysis of \adapt\ (i.e., proof of Theorem~\ref{thm:adaptiveMain}).
\end{description}
This concludes the proof. \qed

\end{proof}
}

Here, we present the formal proof of the key lemma: 
\newline

\noindent \textbf{Lemma~\ref{lemma:k-sizes}:}
Items of size up to $\frac{2}{k-1+\sqrt{(k-1)^2+4}}$ are $\{k-c\}-small$, for $k > c\geq 3$. 

\begin{proof}

It suffices to show that 
\begin{align}
1-\gamma_{k-3}\geq \frac{2}{k-1+\sqrt{(k-1)^2+4}} \label{eqn:k-sizes_lemma}
\end{align}

That is, since $\gamma_k$ is increasing in $k$, we have $1-\gamma_{k^*}>1-\gamma_k$ whenever $k^*<k$. Thus $1-\gamma_{k-c}>1-\gamma_{k-3}$, when $c>3$ and $k-c$, $k-3$ feasible, i.e., $k>c$.

Notice that we only care about showing the lemma for $k>c$, since we only need to use it for cases where $k-c$ is \textit{``feasible''}, in the sense that $k-c\geq 1$ (as it is for $k$ in general).

By expanding equation~(\ref{eqn:k-sizes_lemma})  we get that it must be:
\begin{align}
    1-\frac{(k-3)-2+\sqrt{(k-3)^2+4}}{2(k-3)} &\geq\frac{2}{k-1+\sqrt{(k-1)^2+4}} \Rightarrow \nonumber \\
    \frac{(k-3)+2-\sqrt{(k-3)^2+4}}{2(k-3)} &\geq\frac{2}{k-1+\sqrt{(k-1)^2+4}} \Rightarrow \nonumber
    \end{align}
    
    \begin{align}
     \frac{\Big [(k-3)+2-\sqrt{(k-3)^2+4}\Big ]  \Big [ k-1+\sqrt{(k-1)^2+4} \Big ] - 4(k-3)}{2(k-3)\Big [ k-1+\sqrt{(k-1)^2+4} \Big ]} \geq 0 \label{eqn:k-sizes_fraction}
\end{align}

Notice that the denominator of the fraction in equation~\ref{eqn:k-sizes_fraction} is always positive for feasible $k$, i.e., for $k>3$.

Thus it suffices to show that for $k>3$
\begin{align}
     \Big [(k-3)+2-\sqrt{(k-3)^2+4}\Big ] \Big [  k-1+\sqrt{(k-1)^2+4} \Big ] - 4(k-3) &\geq 0 \label{eqn:k-sizes_inequality}
\end{align}

Set $a=k-3>0$, so that $k-1=a+2$. Define:
\begin{align}
E&=\Big[(k-3)+2-\sqrt{(k-3)^2+4}\Big]\Big[k-1+\sqrt{(k-1)^2+4}\Big]-4(k-3). \label{eq:definitionE}
\end{align}
Using the conjugate, we obtain:
\begin{align}
\Big((a+2)-\sqrt{a^2+4}\Big)\Big((a+2)+\sqrt{a^2+4}\Big)
&=(a+2)^2-(a^2+4)=4a, \label{eq:conjugate}
\end{align}
which implies:
\begin{align}
(a+2)-\sqrt{a^2+4}
&=\frac{4a}{(a+2)+\sqrt{a^2+4}}
=\frac{4(k-3)}{(k-1)+\sqrt{(k-3)^2+4}}. \label{eq:rewrite}
\end{align}
Substituting \eqref{eq:rewrite} into \eqref{eq:definitionE}, we get:
\begin{align}
E
&=4(k-3)\left(\frac{(k-1)+\sqrt{(k-1)^2+4}}{(k-1)+\sqrt{(k-3)^2+4}}-1\right) \label{eq:substitution}\\
&=4(k-3)\frac{\sqrt{(k-1)^2+4}-\sqrt{(k-3)^2+4}}{(k-1)+\sqrt{(k-3)^2+4}}. \label{eq:simplified}
\end{align}
Since $k>3$, we have:
\begin{align}
4(k-3)>0,\quad (k-1)+\sqrt{(k-3)^2+4}>0,\quad 
\sqrt{(k-1)^2+4}>\sqrt{(k-3)^2+4}, \label{eq:positivity}
\end{align}
and therefore $E\ge 0$, which proves the inequality. \qed

\end{proof}

\par\bigskip
\pjmargincomment{Put this here.}\pjedit{\noindent\textbf{AI Disclosure.} After writing the full version of the paper, chatGPT was used to assist with mathematical proofreading, minor corrections and streamlining the writing.}

\end{document}